\documentclass[11pt]{article}

\usepackage{amsmath,amssymb,amsthm,amsfonts,latexsym,bbm,xspace}
\usepackage{graphicx,float,mathtools,braket,enumitem,booktabs}
\usepackage{mathdots,soul,forest,aligned-overset,thm-restate}

\usepackage[useregional]{datetime2}
\DTMusemodule{english}{en-US}

\usepackage[letterpaper,margin=1in]{geometry}

\usepackage[pagebackref,colorlinks,citecolor=blue,,linkcolor=magenta,bookmarks=true]{hyperref}
\usepackage[nameinlink,capitalize]{cleveref}

\usetikzlibrary{patterns}
\usetikzlibrary{matrix}
\usetikzlibrary{arrows.meta}
\usepackage{siunitx}

\newcommand{\eps}{\epsilon}

\renewcommand{\backref}[1]{}

\renewcommand{\backrefalt}[4]{%
\ifcase #1 %
\or
[p.\ #2]%
\else
[pp.\ #2]%
\fi}

\newtheorem{theorem}{Theorem}[section]
\newtheorem{lemma}[theorem]{Lemma}

\newtheorem{proposition}[theorem]{Proposition}
\newtheorem{corollary}[theorem]{Corollary}

\newtheorem{question}[theorem]{Question}

\newtheorem{fact}[theorem]{Fact}
\crefname{fact}{Fact}{Facts}

\theoremstyle{definition}
\newtheorem{definition}[theorem]{Definition}

\newtheorem{remark}[theorem]{Remark}

\renewcommand{\Pr}{\mathop{\bf Pr\/}}
\newcommand{\E}{\mathop{\bf E\/}}

\DeclareMathOperator{\tr}{tr}

\DeclareMathOperator{\poly}{poly}

\newcommand{\yes}{\text{yes}}
\newcommand{\no}{\text{no}}

\newcommand{\reals}{\mathbb R}
\newcommand{\complex}{\mathbb C}
\newcommand{\nats}{\mathbb N}

\newcommand{\C}{\complex}
\newcommand{\N}{\nats}

\newcommand{\class}[1]{\ensuremath{\mathsf{#1}}\xspace}
\mathchardef\mhyphen="2D %

\newcommand{\PTIME}{\class{P}}
\newcommand{\BPP}{\class{BPP}}
\newcommand{\NP}{\class{NP}}
\newcommand{\coNP}{\class{coNP}}
\newcommand{\MA}{\class{MA}}

\newcommand{\EXP}{\class{EXP}}
\newcommand{\NEXP}{\class{NEXP}}

\newcommand{\IP}{\class{IP}}

\newcommand{\PSPACE}{\class{PSPACE}}

\newcommand{\BQP}{\class{BQP}}
\newcommand{\QCMA}{\class{QCMA}}
\newcommand{\QMA}{\class{QMA}}
\newcommand{\QMAk}[1][k]{\class{QMA(#1)}}
\newcommand{\QMAtwo}{\QMAk[2]}
\newcommand{\QIP}{\class{QIP}}
\newcommand{\QIPk}[1][k]{\class{QIP(#1)}}
\newcommand{\QRG}{\class{QRG}}
\newcommand{\QRGk}[1][k]{\class{QRG(#1)}}
\newcommand{\QRGone}{\QRGk[1]}
\newcommand{\QRGtwo}{\QRGk[2]}
\newcommand{\RG}{\class{RG}}

\newcommand{\QMAM}{\class{QMAM}}

\newcommand{\CH}{\class{CH}}
\newcommand{\PP}{\class{PP}}

\newcommand{\EXPH}{\class{EXPH}}

\newcommand{\PH}{\class{PH}}
\newcommand{\Sigmai}[1][i]{\class{\Sigma^{p}_{#1}}}
\newcommand{\Pii}[1][i]{\class{\Pi^{p}_{#1}}}

\newcommand{\QCPH}{\class{QCPH}}
\newcommand{\QCSigmai}[1][i]{\class{QC\Sigma_{#1}}}
\newcommand{\QCPii}[1][i]{\class{QC\Pi_{#1}}}

\newcommand{\QPH}{\class{QPH}}
\newcommand{\QSigmai}[1][i]{\class{Q\Sigma_{#1}}}
\newcommand{\QPii}[1][i]{\class{Q\Pi_{#1}}}
\newcommand{\pureQPH}{\class{pureQPH}}

\newcommand{\QEPH}{\class{QEPH}}
\newcommand{\entQSigmai}[1][i]{\class{QE\Sigma_{#1}}}
\newcommand{\entQPii}[1][i]{\class{QE\Pi_{#1}}}
\newcommand{\entQDeltai}[1][i]{\class{QE\Delta_{#1}}}

\newcommand{\QMAH}{\class{QMAH}}
\newcommand{\QCMAH}{\class{QCMAH}}

\newcommand{\StwoP}{\class{S_{2}P}}

\newcommand{\distributionPH}{\mixedPH}
\newcommand{\distributionQCPH}{\mathsf{DistributionQCPH}}
\newcommand{\mixedPH}{\class{DistributionPH}}

\newcommand{\mixedSigmai}[1][i]{\mathsf{Distribution\Sigma_{#1}}}

\newcommand{\prob}[1]{\textsc{#1}}

\newcommand{\calA}{\mathcal{A}}
\newcommand{\calB}{\mathcal{B}}
\newcommand{\calC}{\mathcal{C}}
\newcommand{\calD}{\mathcal{D}}

\newcommand{\calH}{\mathcal{H}}

\newcommand{\calM}{\mathcal{M}}

\newcommand{\calX}{\mathcal{X}}
\newcommand{\calY}{\mathcal{Y}}

\newcommand{\eg}{e.g.\xspace}

\newcommand{\abs}[1]{\left\lvert #1 \right\rvert}

\newcommand{\qreg}[1]{\mathsf{#1}}

\title{The Entangled Quantum Polynomial Hierarchy Collapses}

 \author{
     Sabee Grewal\thanks{Email: \href{mailto:sabee@cs.utexas.edu}{sabee@cs.utexas.edu}, \href{mailto:yirka@utexas.edu}{yirka@utexas.edu}. The University of Texas at Austin, USA. Supported via Scott Aaronson by a Vannevar Bush Fellowship from the US Department of Defense, the Berkeley NSF-QLCI CIQC Center, a Simons Investigator Award, and the Simons “It from Qubit” collaboration.}
     \and 
     Justin Yirka\footnotemark[1] 
 }

\date{}

\begin{document}

\maketitle

\begin{abstract}

We introduce the entangled quantum polynomial hierarchy $\QEPH$ as the class of problems that are efficiently verifiable given alternating quantum proofs that may be entangled with each other.
We prove $\QEPH$ collapses to its second level. In fact, we show that a polynomial number of alternations collapses to just two.
As a consequence, $\QEPH = \QRGk[1]$, the class of problems having one-turn quantum refereed games, which is known to be contained in $\PSPACE$.
This is in contrast to the \emph{unentangled} quantum polynomial hierarchy $\QPH$, which contains $\QMAtwo$.

We also introduce a generalization of the quantum-classical polynomial hierarchy $\QCPH$ where the provers send probability distributions over strings (instead of strings) and denote it by $\distributionQCPH$. 
Conceptually, this class is intermediate between $\QCPH$ and $\QPH$. 
We prove $\distributionQCPH = \QCPH$, suggesting that only quantum superposition (not classical probability) increases the computational power of these hierarchies. 
To prove this equality, we generalize a game-theoretic result of Lipton and Young (1994) which says that the provers can send distributions that are uniform over a polynomial-size support. 
We also prove the analogous result for the polynomial hierarchy, i.e., $\mixedPH = \PH$.
These results also rule out certain approaches for showing $\QPH$ collapses.

Finally, we show that $\PH$ and $\QCPH$ are contained in $\QPH$, resolving an open question of Gharibian et al.\ (2022).
\end{abstract}

\section{Introduction}\label{sec:intro}

The polynomial hierarchy \cite{meyer1972equivalence,stockmeyer1976polynomial} is a hierarchy of complexity classes that are known to equal $\PTIME$ if and only if $\PTIME = \NP$. 
The hierarchy, denoted by $\PH$, is a natural generalization of efficient proof verification and nondeterminism and plays a central role in complexity theory. 
Given its significance, it is natural to explore quantum generalizations of $\PH$, yet such generalizations remain understudied.

Before discussing quantum polynomial hierarchies, let us first informally define $\PH$. 
Intuitively, $\PH$ is a hierarchy of complexity classes that can solve progressively harder problems, extending beyond both \NP and \coNP.
One can think of $\PH$ as a public debate between Alice and Bob, who take turns presenting polynomial-sized proofs (bit strings) to a referee. 
At the end of the debate, the referee takes the proofs, performs a polynomial-time classical computation, and decides a winner. 

More formally, a problem is in the $k$-th level of the polynomial hierarchy $\Sigmai[k]$ if there is a deterministic polynomial-time verifier $M$ (the referee) that takes proofs $y_1, \dots, y_k$ and satisfies the following conditions. On yes-instances, $\exists y_1 \forall y_2 \exists y_3 \ldots$ such that $M(y_1,\ldots,y_k)=1$, and, on no-instances, $\forall y_1 \exists y_2 \forall y_3 \ldots$ such that $M(y_1,\ldots,y_k)=0$. $\PH$ is comprised of every level $\Sigmai[k]$ for all natural numbers $k$, and it is strongly believed that $\PH$ is infinite.

Gharibian, Santha, Sikora, Sundaram, and Yirka \cite{GSSSY22_qph} studied two quantum generalizations of $\PH$. 
They generalized the class $\QCMA$ to the quantum-classical polynomial hierarchy $\QCPH$, the class of problems for which a quantum verifier can efficiently verify solutions given a constant number of classical proofs from competing provers. Note that this is the same as $\PH$ except the verifier can perform a polynomial-time \emph{quantum} computation.
In the same work, they generalized the class $\QMAtwo$ to the \emph{unentangled} quantum polynomial hierarchy $\QPH$, 
for which the verifier is still quantum, but the proofs are quantum mixed states and promised to be unentangled from each other.
Notably, Gharibian et al.\ did not introduce a hierarchy in which the proofs can be \emph{entangled}, and they did not establish a relationship between $\QPH$ and $\QCPH$ (or even $\QPH$ and $\PH$), leaving it unclear whether or not $\QPH$ was at least as powerful as its classical counterpart.\footnote{While these containments are what one might guess to be true, proving them is nontrivial.} 
More generally, if $\QCPH$ and $\QPH$ are indeed more powerful, it prompts the question of why: is it quantum verification, quantum proofs, unentanglement, or some nuanced combination?

In this work, we address all of these questions related to quantum generalizations of the polynomial hierarchy.
First, we ask (and answer) what problems admit a $\PH$-style protocol where the provers can send potentially entangled proofs. 
We show that this new hierarchy---the entangled quantum polynomial hierarchy ($\QEPH$)---behaves drastically differently from what we believe about $\PH$, $\QCPH$, and $\QPH$.

Second, we prove that $\PH \subseteq \QCPH \subseteq \QPH$, 
confirming the intuitive relationship between these hierarchies. 

Lastly, to understand the power of quantum proofs, we introduce a generalization of $\QCPH$ where the provers send probability distributions over classical proofs and denote the class by $\distributionQCPH$. 
We prove that $\distributionQCPH = \QCPH$, despite the intuition from game theory that mixed strategies are helpful.
Note that the only difference between $\distributionQCPH$ and $\QPH$ is that the proofs in $\QPH$ involve quantum superposition.
Hence, our result establishes that the increased computational power of $\QPH$ comes only from the quantum superposition in the proofs.

\subsection{Our Results}

Our first main result is a characterization of our newly defined hierarchy $\QEPH$ (\cref{def:entqph}) via a collapse to its second level. This collapse is in stark contrast to our belief that $\PH$ is infinite.

\begin{theorem}[Combination of \cref{lem:ent-qph-collapse,thm:qrg-one-equals-pspace}] 
$\QEPH$ collapses to its second level and equals $\QRGone$.
\end{theorem}

This collapse is similar to others known in quantum complexity theory, such as $\QIP = \QIPk[3] = \QMAM$ \cite{kitaev2000parallelization,marriott2005quantum}, 
in which the protocols rely on the prover's ability to entangle their messages.
We further compare $\QEPH$ to other complexity classes involving entangled proofs in \nameref{subsec:related-work}.

We show that $\QEPH$ equals $\QRGone$, the class of problems having one-turn quantum-refereed games.\footnote{The class $\QRG(k)$ and its classical analogue $\RG(k)$ have been numbered differently by different authors. We follow recent conventions
where the provers' and the referee's messages are counted separately.
So, e.g., in $\QRGtwo$ the referee sends one message and then the provers each simultaneously send a message.}
$\QRGone$ involves a game between two competing players that each privately sends a quantum state to a referee, who then performs a polynomial-time quantum computation to determine a winner.
In 2009, Jain and Watrous \cite{jain2009parallel} proved $\QRGone \subseteq \PSPACE$. 
However, it is conjectured that $\QRGone$ is strictly less powerful than $\QRGtwo=\PSPACE$ \cite{gutoski2013parallel,ghosh2023complexity}. 
Yet despite effort, no improved upper bounds on $\QRGone$ have been proven in over a decade. 
We suggest a new approach to improving the upper bound on $\QRGone$ (via the connection to $\QEPH$) in \nameref{sec:openproblems}.

Our collapse result is stronger than stated above.
It is well-known that if one extends $\PH$ to a polynomial number of rounds (rather than a constant number), then the resulting class equals $\PSPACE$ \cite[Theorem 4.11]{arora2009computational}.
In contrast, we show that extending $\QEPH$ to a polynomial number of rounds does not increase the power of the class. 

\begin{theorem}[Informal version of \cref{corr:entQSigmaPoly=2}]
   Even with a polynomial number of rounds, $\QEPH$ collapses to its second level.
\end{theorem}

One interpretation of our collapse result is that allowing provers to entangle their proofs gives them too much opportunity to cheat. Hence, receiving a single proof from each prover is just as useful as receiving many entangled proofs. 

Before this work, it was unclear how the quantum polynomial hierarchies compared to one another, and if $\QPH$ even contained $\PH$. 
In our second result, we establish the following containments between the quantum and classical hierarchies, resolving an open question of Gharibian et al. \cite{GSSSY22_qph}.

\begin{theorem}[Restatement of \cref{thm:direct-proof-ph-prodqph}]
    $\PH \subseteq \QCPH \subseteq \QPH$.
\end{theorem}

We emphasize that even $\PH \subseteq \QPH$ is not obvious.
Placing restrictions on the provers can sometimes increase computational power, as was the case in, e.g., the recent results showing that $\QMA^+ = \QMAtwo^+ = \NEXP$ \cite{jeronimo2023power, bassirian2023quantum}. 
Meanwhile, the permissiveness of $\QEPH$, where we allow the provers to entangle their proofs, seems to yield a weaker class than $\QPH$.

In our third result, we show that the power of $\QCPH$ does not change if the provers are allowed to send probability distributions (instead of a fixed classical proof). 

\begin{theorem}[Restatement of \cref{cor:distqcph=qcph}]\label{thm:intro-distqcph=qcph}
   $\distributionQCPH = \QCPH.$
\end{theorem}

Our motivation for studying $\distributionQCPH$ is to better understand the power of quantum proofs. 
In particular, let $\pureQPH$ be the same as $\QPH$ except the quantum proofs are pure states rather than mixed states.\footnote{It is easy to see that $\QPH \subseteq \pureQPH$ since the provers can send purifications of their mixed proofs.\label{footnote:pureQPH}} 
Then the \emph{only} difference between $\pureQPH$ and $\distributionQCPH$ is that the former involves proofs that are quantum superpositions over bit strings while the latter involves proofs that are classical distributions over bit strings. 
Yet $\distributionQCPH = \QCPH$ is in the counting hierarchy \cite{GSSSY22_qph}, and $\pureQPH$ contains $\QMAtwo$ and is contained in $\EXP^\PP$ \cite{agarwal2023quantum}. Conceptually, our result says that any increase in computational power only comes from the quantum superposition in the proofs.

\cref{thm:intro-distqcph=qcph} also goes through for $\PH$. 

\begin{theorem}[Restatement of \cref{thm:mixedph=ph}]\label{thm:intro-mixedph}
   $\distributionPH = \PH.$
\end{theorem}

An easy consequence of our result is that $\distributionPH$ collapses if and only if $\PH$ collapses.\footnote{$\class{Distribution}\calC$ is not to be confused with the notation $\class{Dist}\calC$, which has been used in average-case complexity theory, \eg{} $\class{DistNP}$ and $\class{DistPH}$.
Also, similar names like \prob{Stochastic SAT} or \prob{Probabilistic QBF} have appeared in the study of randomized quantifiers. These are more similar to the Arthur-Merlin ($\class{AM}$) hierarchy. 
}
Therefore, any attempts to collapse $\QCPH$, $\QPH$, or $\pureQPH$ must not collapse $\distributionPH$, and so \cref{thm:intro-mixedph} rules out some approaches to collapsing these hierarchies.
In particular, one line of attack to showing $\QMAtwo = \NEXP$ is to show that the $\forall$ quantifier in $\QSigmai[3]$ does not add any computational power, because $\QMAtwo \subseteq \QSigmai[3] \subseteq \NEXP$ \cite{GSSSY22_qph}. \cref{thm:intro-distqcph=qcph,thm:intro-mixedph} are evidence that this line of attack won't work, since showing the analogous result for $\distributionPH$ would collapse the polynomial hierarchy.

We give a graphical description of our results and the quantum polynomial hierarchy landscape in \cref{fig:hierarchy-relations}.

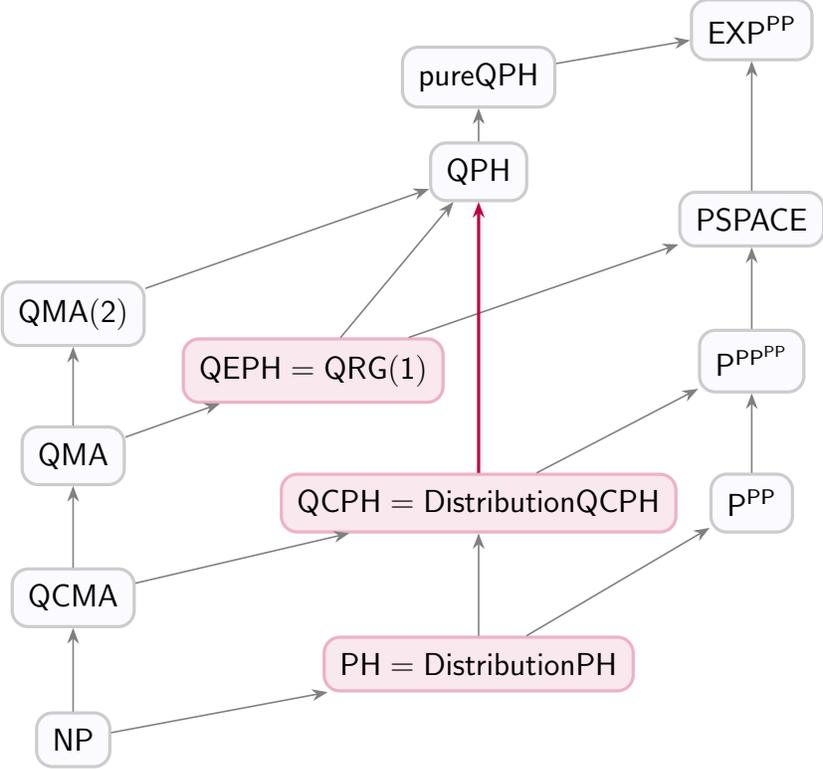
\begin{figure}[ht]
\centering
\begin{tikzpicture}[scale=1.1]
\tikzset{inner sep=0,outer sep=3}
\tikzstyle{a}=[inner sep=6pt, inner ysep=6pt,outer sep=0.5pt,
draw=black!20!white, fill=blue!2!white, very thick, rounded corners=6pt, align=center]
\begin{scope}[yscale=1.145]
\tikzstyle{b}=[inner sep=6pt, inner ysep=6pt,outer sep=0.5pt,
draw=purple!30!white, fill=purple!9!white, very thick, rounded corners=6pt, align=center]
\large
\node[a] (NP) at (-6.2,1) {$\NP$};
\node[a] (QCMA) at (-6.2,2.5) {$\QCMA$};
\node[a] (QMA) at (-6.2,4) {$\QMA$};
\node[a] (QMAtwo) at (-6.2,5.5) {$\QMAtwo$};
\node[b] (entQPH) at (-3.3,4.9) {$\QEPH = \QRGone$};
\node[b] (PH) at (-1.3,1.8) {$\PH = \distributionPH$};
\node[b] (QCPH) at (-1.3,3.5) {$\QCPH = \distributionQCPH$};
\node[a] (prodQPH) at (-1.3,7) {$\QPH$};
\node[a] (pureQPH) at (-1.3,8) {$\pureQPH$};
\node[a] (PSPACE) at (2,6.5) {$\PSPACE$};
\node[a] (EXPH) at (2,8.5) {$\EXP^\PP$};
\node[a] (PPPPP) at (2,5.0) {$\PTIME^{\PP{^\PP}}$};
\node[a] (PPP) at (2,3.5) {$\PTIME^\PP$};
\end{scope}
\path[-{Stealth[length=6pt]},line width=.6pt,gray]
(NP) edge (QCMA)
(NP) edge (PH)
(PH) edge (QCPH)
(PH) edge (PPP)
(QCPH) edge (PPPPP)
(PPPPP) edge (PSPACE)
(PPP) edge (PPPPP)
(QCMA) edge (QMA)
(QMA) edge (QMAtwo)
(QMAtwo) edge (prodQPH)
(PSPACE) edge (EXPH)
(prodQPH) edge (pureQPH)
(pureQPH) edge (EXPH)
(QCMA) edge (QCPH)
(QMA) edge (entQPH)
(entQPH) edge (prodQPH)
(entQPH) edge (PSPACE)
;
\path[-{Stealth[length=6pt]},line width=1.2pt,purple]
(QCPH) edge (prodQPH)
;
 (BQP);
\end{tikzpicture}
    \caption{(Color) The quantum polynomial hierarchy landscape in light of our work. The containments and complexity classes shown in grey were previously known, and the containments and complexity classes in red are contributions of this work.}
\label{fig:hierarchy-relations}
\end{figure}

\subsection{Main Ideas}\label{ssec:techniques}

Let us begin by explaining $\QEPH$ on an intuitive level (see \cref{def:entqph} for a formal definition). 
$\QEPH$ can be thought of as a constant-round non-interactive game between two competing provers, Alice and Bob, who take turns sending quantum registers, i.e., collections of qubits, to a verifier. Alice and Bob are allowed to entangle their own quantum registers across turns.
The verifier then performs a polynomial-time quantum computation, measures a fixed output qubit in the computational basis, and, if the verifier sees $1$, they accept (Alice wins), and reject otherwise (Bob wins). 
$\QEPH$ contains the decision problems for which Alice always wins with high probability on yes-instances and Bob always wins with high probability on no-instances. 
We note that in this game the moves are \emph{public}, which means that Alice knows the state of the quantum registers sent by Bob and vice versa.
See \cref{remark:public} for further discussion of public vs. private moves in a quantum world.

To highlight the key technique in our proof that $\QEPH$ collapses (\cref{lem:ent-qph-collapse}), we explain how to simulate the third level of $\QEPH$, denoted by $\entQSigmai[3]$, inside of the second level $\entQSigmai[2]$. The proof for higher levels proceeds by induction.
As we will explain formally in \cref{sec:prelim-entqph}, a $\entQSigmai[i]$ protocol can be written as an optimization problem with a value equal to the probability the verifier accepts when both players use optimal strategies.
In particular, Alice selects proofs that maximize the probability of the verifier accepting, while Bob selects proofs to minimize that probability.
For $\entQSigmai[3]$, given a problem instance in which the verifier's action is encoded by an observable $R$, the corresponding optimization problem is
\begin{equation*}
    \max_{\rho_1 \in \mathbf{D}(\calX_1)} \min_{\sigma \in \mathbf{D}(\calY)} \max_{\rho_2 \in \calA}\, \tr\left( R \left( \rho_2 \otimes \sigma\right) \right),  
\end{equation*}
where $\mathbf{D}(\calH)$ denotes the set of density operators on the Hilbert space $\calH$
and $\calA \coloneqq \{ \rho \in \mathbf{D}(\calX_1 \otimes \calX_2) \mid \tr_{\calX_2}(\rho) = \rho_1\}$.
The restriction of the second maximization to the set $\calA$ is 
to enforce that Alice's second move is consistent with her first.

A straightforward analysis shows that when focusing on the inner two operators, a min-max theorem applies, allowing us to swap the ordering of the inner minimization and maximization. Then, because we allow entangled states, we can combine the two sequential maximization operators into one, leaving an optimization problem corresponding to a two-round protocol.
Notably, both the three-round and two-round protocols are over the same input and verifier, so the reduction does not increase the problem size or change the error parameters.

It is natural to ask why 
our technique does not also collapse $\PH$. In short, the above approach fails immediately, since, for one, our collapse theorem relies on the fact that Alice and Bob are choosing quantum proofs from compact and convex sets (see \cref{fact:density-matrices-compact-convex,fact:partial-trace-compact-convex}). 
In contrast, the set of classical strings is neither compact nor convex. 

To show that $\QEPH = \QRGone$, we build on a previous characterization of Gharibian et al.\ \cite{GSSSY22_qph} where they showed that the second level of the \emph{unentangled} quantum polynomial, denoted by $\QSigmai[2]$, equals $\QRGone$. 
We extend their result in \cref{prop:second-levels-equal} to show that $\entQSigmai[2] = \QSigmai[2]=\QRGone$, which yields our characterization that $\QEPH = \entQSigmai[2]=\QRGone$. 
$\entQSigmai[2] = \QSigmai[2]$ because, after two turns, each prover has only sent a single proof, so there's no distinction yet to be made between the entangled versus entangled hierarchies.

\paragraph{}
We now turn to the containment $\QCPH \subseteq \QPH$, which are both defined formally in \cref{subsec:prelim-qph}.
In $\QCPH$, the verifier receives classical proofs, whereas the proofs in $\QPH$ are unentangled quantum mixed states.
One na\"ive approach to simulating $\QCPH$ inside of $\QPH$---which does not work---is for the verifier to immediately measure the quantum proofs to get classical strings and then run the $\QCPH$ verification protocol. 
The reason this fails is that the dishonest prover (i.e., the player without a winning strategy) can cheat by sending a quantum state, rather than a classical proof. 
In more detail, while the honest prover has perfect knowledge of the quantum states sent by the dishonest prover, they do not know which particular classical strings the verifier will observe upon measurement, making it unclear what their response should be.
The definition of $\QCPH$ guarantees the correct player has an effective response conditioned on any particular proof sent from the other player, but this does not generally guarantee the correct player can succeed against a mixture of potential moves.
Unfortunately, the equilibrium point of a zero-sum game which allows for such mixed moves will generally be mixed, rather than pure.

To overcome this, we simulate the $i$-th level of $\QCPH$ in the $2ki$-th level of $\QPH$, for some constant $k$. 
We ask the provers to send $k$ copies of each of the proofs they would send in the $\QCPH$ protocol, which increases the number of turns by a factor of $2k$. 
Using the groups of $k$ proofs, we give a simple test to ensure that no player cheats, which works as follows. 
Measure each of the $k$ proofs in the standard basis. 
If the outcomes are all equal, then the test passes, and, otherwise, the test fails. 
We prove that this is enough to force the provers to send computational basis states with high probability. 

We remark that this bears some similarity to other protocols involving unentanglement. 
Harrow and Montanaro \cite{harrow2013testing} used unentanglement to force Merlin to send $k$-partite states, and, recently, Jeronimo and Wu \cite{jeronimo2023power} use unentanglement to force Merlin to send many copies of (approximately) the same quantum state. 
Both of these results fundamentally rely on the swap test, which tests for equality between two quantum states \cite{buhrman2001quantum}. 
In a similar fashion, we use unentanglement to force the provers to send standard basis states, i.e., classical strings. 
With that, we design a simulation of any $\QCPH$ protocol inside of $\QPH$. 

\paragraph{}

Finally, we discuss our proof that $\distributionQCPH = \QCPH$ (the same techniques will also show $\distributionPH = \PH$). 
In $\distributionQCPH$, the provers send probability distributions over polynomial-length classical proofs. 
For classical proofs of length $m$, these distributions can have a support of size exponential in $m$. 
The key lemma for this result says that the provers can send \emph{much simpler} distributions without changing the acceptance probability of the verifier too much. 
In particular, we prove that the distributions sent by the provers can be \emph{uniform} over $\poly(m)$ many classical proofs and, even with this simplification, the acceptance probability of the verifier will change by at most a small constant.  
This simplification lemma (\cref{lem:lipton}) generalizes a result due to Lipton and Young \cite{lipton1994simple} and Althöfer \cite{ALTHOFER1994339}, where they showed that such a simplification works in the special case of a one-turn two-player game. 
Our contribution is to generalize their result to the case where the two players alternate sending distributions any constant number of times.

With the simplification lemma, one can prove $\distributionQCPH \subseteq \QCPH$ as follows. 
To send a distribution in $\QCPH$, the provers send every classical string that is in the support of their distribution. By our simplification lemma, there are only a polynomial number of such strings, so all of them can be sent in a polynomially-sized classical proof.  
Then, since the simplified distributions are uniform, the verifier can randomly sample one of the strings uniformly at random.
The other direction $\distributionQCPH \supseteq \QCPH$ follows from the same techniques that prove $\QCPH \subseteq \QPH$.

\subsection{Related and Concurrent Work}\label{subsec:related-work}
Early efforts to define quantum hierarchies include \cite{yamakami2002quantum,gharibian2012hardness}.

We choose to use alternating $\exists$ and $\forall$ quantifiers to define $\QEPH$ (as was the case for $\QCPH$ and $\QPH$ in \cite{GSSSY22_qph}).
In addition to a quantifier definition, $\PH$ can be \emph{equivalently} defined in the oracle model via constant-height towers of the form $\NP^{\NP^{\NP^{\ldots}}}$.
The oracular definition gives rise to natural definitions of quantum polynomial hierarchies, some of which have been studied recently.
Vinkhuijzen \cite{V18_BQPHthesis} and Aaronson, Ingram, and Kretschmer \cite{AIK22_acrobatics} study the ``$\QMA$ hierarchy'', $\mathsf{QMAH}$, which consists of constant-depth towers of the form $\QMA^{\QMA^{\QMA^{\ldots}}}$.\footnote{Vinkhuijzen only allows recursive queries to $\QMA$, whereas Aaronson, Ingram, and Kretschmer allow recursive queries to ${\mathsf{PromiseQMA}}$.}  
\cite[Theorem 5]{V18_BQPHthesis} shows that $\mathsf{QMAH}$ is contained in the counting hierarchy $\CH$, while the best upperbounds for the quantifier-based hierarchies, $\QEPH$ and $\QPH$, are $\PSPACE$ and $\EXP^\PP$, respectively. 

The method of showing equivalence between the quantifier-based and oracle-based definitions of $\PH$ does not appear to carry over to $\QEPH$, $\QPH$, or even $\QCPH$.
This seems related to the inability to ``pull quantumness out of a quantum algorithm'' as we can for randomness from randomized algorithms \cite{AIK22_acrobatics} as well as a lack of study of quantum oracle machines.
We further discuss questions regarding $\QMAH$ vs. $\QEPH$  in \nameref{sec:openproblems}.

There are several quantum complexity classes that involve provers sending possibly entangled proofs to a quantum polynomial-time verifier. 
We do not attempt to survey them here, but, for convenience, we summarize quantum complexity classes involving entangled proofs (and their classical counterparts) in \cref{table:classes}.

Our work on $\distributionPH$ builds on previous game-theoretic characterizations in complexity theory (see e.g., \cite{feigenbaum1995game}).
$\PH$-style classes involve a debate with public communication (perfect information), and a non-interacting, passive referee.
$\RG$-style classes involve private communication (imperfect information) with provers sending particular strings to the referee (perfect recall).
A consequence of imperfect information is that the players must model their competitor's moves as probability distributions (mixed strategies) because they are never sure which move is made.
Our class $\distributionPH$ fits into this framework in a nuanced way.
Specifically, the distributions sent are public (similar to $\PH$); they represent a mixture of pure moves (similar to $\RG$); but, uniquely, the provers don't know which string will be sampled by the referee (reminiscent of imperfect recall). 
This is a novel game-theoretic model, and as we discuss further in \cref{sec:mixed}, it is naturally motivated by a game of quantum mixed states sent to a non-interacting referee.

Finally, the concurrent and independent work of Agarwal, Gharibian, Koppula, and Rudolph \cite{agarwal2023quantum} also studies generalizations of the polynomial hierarchy.
They prove $\QCPH \subseteq \pureQPH$, which is similar to our \cref{thm:direct-proof-ph-prodqph} that $\QCPH \subseteq \QPH$.
Since $\QPH\subseteq \pureQPH$ is straightforward (the provers send purifications of their proofs), our \cref{thm:direct-proof-ph-prodqph} implies $\QCPH\subseteq \pureQPH$.
In this sense, our containment is stronger.
However, their containment has the nice (and nontrivial) feature that the $k$-th level of $\QCPH$ is contained in the $k$-th level of $\pureQPH$, whereas our containment requires blowing up to the $ck$-th level of $\QPH$ for a constant integer $c$. 
Besides this,
Agarwal et al. contribute several more results including a theorem that if $\QCSigmai=\QCPii$ then $\QCPH$ collapses (see also \cite{falor2023collapsible}); a Karp-Lipton style result that $\QCMA \subseteq \class{BQP/mpoly}$ implies $\QCPH$ collapses; a new upper bound $\QPH\subseteq \pureQPH \subseteq \EXP^{\PP}$, improving on the previous upper bound of $\EXPH$; and a method for one-sided error-reduction of $\pureQPH$.

\begin{table}[ht]
    \centering
    \caption{Complexity classes characterizing proof verification that are related to $\QEPH$. ``C'' means classical and ``Q'' means quantum. For every class below, multiple provers are always competing, and, for multi-round quantum protocols, the quantum proofs can be entangled across rounds. Public means that the provers have full knowledge of their opponent's previous turns.}
    \label{table:classes}
    \begin{tabular}
    {
        l
        S[table-format = 2]
        S[table-format = 2]
        S[table-format = 2]
        S[table-format = 2]
        S[table-format = 2]
        S[table-format = 2]
        S[table-format = 2]
        }
        \toprule
              & {\# of } & {\# of}   & {}       & {}         & {Interaction}            & {Public or}      &  {} \\
        Class & {Rounds} & {Provers} & {Proofs} & {Verifier} & {from referee?} & {Private} & {Equals} \\
        \midrule
        {$\NP$} & {1} & {1} & {C} & {C} & {no} & {N/A} & {} \\
        {$\QMA$} & {1} & {1} & {Q} & {Q} & {no} & {N/A} & {} \\
        {$\IP$} & {poly} & {1} & {C} &{C} & {yes} & {N/A} & {$\PSPACE\,$ \cite{shamir1992ip}} \\
        {$\QIP\mathsf{(3)}$} & {3} & {1} & {Q} &{Q} & {yes} & {N/A} & {$\PSPACE\,$\cite{jain2011qip}} \\
        \midrule 
        {$\PH$} & {const} & {2} & {C} & {C} & {no} & {pub.} & {} \\
        {$\QEPH$} & {const} & {2} & {Q} & {Q} & {no} & {pub.} & {$\QRGone$ [This Work]} \\
        {$\RG(1)$} & {1} & {2} & {C} &{C} & {no} & {priv.} & {$\mathsf{S_2P}$\cite{ALTHOFER1994339,lipton1994simple}} \\
        {$\RG(2)$} & {2} & {2} & {C} &{C} & {yes} & {priv.} & {$\PSPACE\,$\cite{feige1997making}} \\
        {$\RG$} & {poly} & {2} & {C} &{C} & {yes} & {priv.} & {$\EXP\,$\cite{feige1997making}} \\
        {$\RG\mathsf{(pub)}$} & {poly} & {2} & {C} &{C} & {yes} & {pub.} & {$\PSPACE\,$\cite{feige1997making}} \\
        {$\QRG(1)$} & {1} & {2} & {Q} &{Q} & {no} & {priv.} &  \\
        {$\QRG(2)$} & {2} & {2} & {Q} &{Q} & {yes} & {priv.} & {$\PSPACE\,$\cite{gutoski2013parallel}} \\
        {$\QRG$} & {poly} & {2} & {Q} &{Q} & {yes} & {priv.} & {$\EXP\,$\cite{gutoski2007toward}} \\
        \bottomrule
    \end{tabular}
\end{table}

\section{Preliminaries}\label{sec:prelims}

We introduce notation, definitions, and background that are central to our results. 
For the most part, we assume familiarity with common concepts and classes in quantum and classical complexity theory as well as quantum computing and quantum information.
For a thorough discussion of these topics, see \cite{arora2009computational, watrous2018theory, kitaev2002classical, nielsen2002quantum}. 

We will need the following version of Hoeffding's inequality. 

\begin{fact}[Hoeffding’s inequality]\label{fact:hoeffding}
Let $X_1,\dots, X_n$ be independent random variables subject to $a_i \le X_i \le b_i$ for all $i$. 
Let $X = \sum_{i=1}^n X_i$ and let $\mu = \E[X]$. 
Then it holds that
\[
\Pr[X - \mu \geq t] \le \exp\left(- \frac{2t^2}{\sum_{i=1}^n (b_i - a_i)^2} \right)
\]
and 
\[
\Pr[X - \mu \leq -t] \le \exp\left(- \frac{2t^2}{\sum_{i=1}^n (b_i - a_i)^2} \right).
\]
\end{fact}

\subsection{Quantum Information}
A \textit{quantum register} refers to a collection of qubits. Associated with each register is a complex Hilbert space, and the state of a quantum register is described by a Hermitian, positive semi-definite matrix with trace one called a \textit{density matrix}.  
We denote the set of $n$-qubit density matrices by $\mathbf{D}(n)$, and the sets of linear operators and density matrices on a complex Hilbert space $\calH$ by $\mathbf{L}(\calH)$ and $\mathbf{D}(\calH)$, respectively. 

For two quantum registers $(\qreg{X}, \qreg{Y})$ with Hilbert spaces $\calX$ and $\calY$, the combined space is the tensor product space $\calX \otimes \calY$.   
The partial trace $\tr_{\calY}: \mathbf{L}(\calX \otimes \calY) \to \mathbf{L}(\calX)$ is the unique linear map that satisfies $\tr_\calY(A \otimes B) = \tr(A) B$ for all $A \in \mathbf{L}(\calX)$ and $B \in \mathbf{L}(\calY)$. 
If the compound register $(\qreg{X}, \qreg{Y})$ is in the state $\rho \in \mathbf{D}(\calX \otimes \calY)$, then the state of register $\qreg{X}$ is $\tr_\calY(\rho) \in \mathbf{D}(\calX)$.
That is, operationally speaking, the partial trace is the act of ignoring (or discarding) a quantum register.
We note that the partial trace $\tr_\calX$ can be defined similarly, and, in general, the context in which the partial trace is used should clarify which spaces are being ``traced out''.  

A \textit{quantum measurement} of a quantum register is described by a finite collection of Hermitian, positive semi-definite matrices that sum to identity. Let $\qreg{X}$ be a quantum register with Hilbert space $\calX$ whose state is described by $\rho$. Let $\calM = \{E_i \mid i \in \Sigma\}$ be a quantum measurement, where $\Sigma$ is a finite alphabet. Upon measuring $\qreg{X}$ with $\calM$, we observe $i \in \Sigma$ with probability $\tr(E_i \rho)$. 

\subsection{A Min-Max Theorem}
To prove our collapse theorem, we use a weaker version of Sion's min-max theorem. 

\begin{theorem}[A weaker version of Sion's min-max theorem \cite{sion1958general}]\label{fact:min-max}
Let $\calX$ and $\calY$ be complex Euclidean spaces, let $\calA \subseteq \calX$ and $\calB \subseteq \calY$ be convex and compact subsets, and let $f: \calA \times \calB \to \reals$ be a bilinear function. Then 
\[
\max_{a \in \calA} \min_{b \in \calB} f(a, b) = \min_{b \in \calB} \max_{a \in \calA} f(a,b).
\]
\end{theorem}

It is a well-known fact that the space of density matrices is compact and convex.
\begin{fact}[{\cite[Chapter 1]{watrous2018theory}}]\label{fact:density-matrices-compact-convex}
Let $\mathbf{D}(\calH)$ be the set of density matrices on a complex Hilbert space $\calH$. $\mathbf{D}(\calH)$ is compact and convex.
\end{fact}

It is critical for us that, even if we impose partial trace constraints on the set of density matrices, the set remains compact and convex. We include a proof for completeness.

\begin{fact}\label{fact:partial-trace-compact-convex}
Let $\qreg{X}, \qreg{Y}$ be two quantum registers with Hilbert spaces $\calX$ and $\calY$, respectively, and let $\mathbf{D}(\calX \otimes \calY)$ be the corresponding set of density operators. Let $\rho^\prime \in \mathbf{D}(\calX)$ be some fixed density matrix.
Then the set 
\[
\mathbf{S} = \{ \rho \in \mathbf{D}(\calX \otimes \calY) \mid \tr_\calY(\rho) = \rho^\prime \}
\]
is compact and convex.
\end{fact}
\begin{proof}
   Let $\rho_1, \rho_2 \in \mathbf{S}$, and define $\sigma \coloneqq \theta \rho_1 + (1-\theta) \rho_2$ for for $\theta \in [0,1]$. Then  
   \begin{align*}
       \tr_\calY(\sigma) 
       &= \tr_\calY(\theta \rho_1 + (1-\theta) \rho_2) \\
       &= \theta \tr_\calY(\rho_1) + (1-\theta) \tr_\calY(\rho_2) && (\text{By the linearity of the partial trace.}) \\
       &= \theta \rho^\prime + (1-\theta) \rho^\prime && (\text{Because $\rho_1, \rho_2 \in \mathbf{S}$.}) \\
       &=  \rho^\prime,  
   \end{align*}
   so $\mathbf{S}$ is convex. 
   
To show that $\mathbf{S}$ is compact, we must show that it is closed and bounded.
Without loss of generality, let $\qreg{X}$ be an $n$-qubit register and $\qreg{Y}$ be an $m$-qubit register.
Then we can identify $\mathbf{S}$ with the vector space $\complex^{4^{n + m}}$ and observe that all entries are bounded in magnitude by $1$. Therefore, $\mathbf{S}$ is bounded.
To see that $\mathbf{S}$ is closed, we need the following definitions.
For $x \in \complex$, define $f_x: \complex^{4^{n+m}} \to \complex$ as 
$f_x(A) = \langle x, Ax \rangle$,
which is continuous because the inner product is continuous; 
define 
$g:\C^{4^{n + m}}\to\mathbb{C}^{4^{n +m}}$ as 
$g(A) = A - A^\dagger,$
which is a polynomial and therefore continuous;
and, finally, define $h:\C^{4^{n+m}}\to\mathbb{C}^{4^n}$ as 
$h(A) = \tr_\calY(A),$ which is a linear map on a finite-dimensional vector space and therefore continuous.
Then 
\[
\mathbf{S} = \bigcap_{x\in \mathbb{C}}f_x^{-1}([0,\infty))\cap g^{-1}(\{0\})\cap h^{-1}(\{\rho^\prime\})\cap\text{tr}^{-1}(\{1\}).
\]
The preimage of a continuous function on a closed set is closed, and the intersection of closed sets is closed. Therefore, $\mathbf{S}$ is closed. 
\end{proof}

\subsection{Previously Studied Hierarchies}\label{subsec:prelim-qph}

Here, we give formal definitions of the polynomial hierarchy $\PH$, 
the quantum-classical polynomial hierarchy $\QCPH$,
and the unentangled quantum polynomial hierarchy $\QPH$, the latter two of which were both introduced by Gharibian et al. \cite{GSSSY22_qph}.
These classes will appear again in \cref{sec:direct-proof-ph-prodqph} when we prove $\QCPH\subseteq \QPH$ and in \cref{sec:mixed} when we prove $\distributionQCPH= \QCPH$.
We defer definitions of our new classes until later, with $\QEPH$ studied in \cref{sec:entQPHcollapseInPSPACE} and $\distributionQCPH$ in \cref{sec:mixed}.

\begin{definition}[$\Sigmai$]
    A language $L$ is in the $i$-th level of the polynomial hierarchy $\Sigmai$ if there exists a polynomial-time deterministic Turing Machine $M$ such that for any $n$-bit input $x$,
    \begin{align*}
        x\in L &\iff  \exists y_1 \forall y_2 \exists y_3 \dots Q_i y_i \text{ such that } M(x,y_1,\dots,y_i)=1 , \\
        x\not\in L &\iff \forall y_1 \exists y_2 \forall y_3 \dots \overline{Q_i} y_i \text{ such that } M(x,y_1,\dots,y_i)=0 ,
    \end{align*}
    where $Q_i$ denotes $\exists$ if $i$ is odd and $\forall$ otherwise, $\overline{Q_i}$ denotes the complement of $Q_i$, and $\abs{y_i} \leq p(n)$ for some fixed polynomial $p$ for all $i$.
\end{definition}

\begin{definition}[The polynomial hierarchy ($\PH$) \cite{stockmeyer1976polynomial}] The Polynomial-time Hierarchy is defined as 
    \[\PH \coloneqq \bigcup_{i=0}^{\infty} \Sigmai.\]
\end{definition}

Note the union which defines $\PH$ is over values of $i$ which are constant, independent of a problem's input size.
Observe also that for all $i$, $\Sigmai[i]\subseteq \Sigmai[i+1]$. Additionally, $\PH$ is closed under complement, in particular because $\class{\overline{\Sigma}_i^p} \subseteq \Sigmai[i+1] \subseteq \PH$. 
The complement of $\class{\overline{\Sigma}_i^p}$ is defined to be $\Pii[i]$, and for all $i$ we have $\Sigmai[i]\subseteq \Pii[i+1] \subseteq \Pii[i+2]$.

The definition of $\PH$ is particularly robust.
The class can be defined equivalently by $\Sigmai[i+1] = \NP^{\Sigmai[i]}$, giving a constant-height tower of $\NP$ oracles.
The model of alternating nondeterministic Turing Machines also can be used to define each level of the hierarchy.
In another direction, the Sipser–Lautemann theorem shows $\BPP \subseteq \Sigmai[2]\cap \Pii[2] \subseteq \PH$ \cite{sipser1983complexity,LAUTEMANN1983215}.
So, natural bounded-error or probabilistic definitions of $\PH$ collapse to the standard, deterministic definition given above. This is also true for oracle definitions, where we know $\MA^{\MA^{\ldots}} = \PH$.

Even a partial survey of results regarding $\PH$ would be impossible to fit here. We finally note that $\Sigmai[i]=\Sigmai[i+1]$ or $\Sigmai[i]=\Pii[i]$ would both ``collapse'' the hierarchy so that $\PH=\Sigmai[i]$. These two events are analogous to $\PTIME=\NP$ or $\NP=\coNP$. Conversely, if $\PH$ collapses to any finite level, it implies analogs of $\PTIME=\NP$ and $\NP=\coNP$ must be true for some degree of nondeterminism, at some level of the hierarchy. So, the strongly-believed conjecture that $\PH$ is not equal to any $\Sigmai[i]$ for fixed $i$ is a generalization of those other strongly-believed conjectures.

The uniform circuit model is standard for quantum complexity classes, so we give the definition below.

\begin{definition}[Polynomial-time uniform family of quantum circuits]
A polynomial-time uniform family of quantum circuits is a family $\{V_n\}_{n \in \N}$ such that there exists a polynomial bounded function $t: \N \to \N$ and a deterministic Turing machine $M$ acting as follows. 
For every $n$-bit input $x$, $M$ outputs in time $t(n)$ a description of a quantum circuit $V_n$, which has a designated output qubit. 
We say $V_n$ accepts when we observe a $1$ upon measuring the designated output qubit in the standard basis. 
\end{definition}

We generally leave the subscript implicit and just write $V$.
Additionally, we often consider a single problem instance defined by an input $x$ for the full length of an analysis.
So instead of writing $V(x,y)$ for input $x$ and proof $y$,
we simply refer to $V(y)$.

As with most quantum complexity classes, we will be working with promise problems. 
Briefly, a promise problem $A$ is a pair of non-intersecting subsets ($A_\yes$, $A_\no$) of $\{0,1\}^*$. A decision problem, or language, is a promise problem where $A_\yes\cup A_\no = \{0,1\}^*$.

We are now ready to define $\QCPH$.

\begin{definition}[$\QCSigmai$ \cite{GSSSY22_qph}] \label{def:qcsigmai}
    A promise problem $L=(L_\yes, L_\no)$ is in $i$-th level of the quantum-classical polynomial hierarchy $\QCSigmai[i]{}\class{(c,s)}$ for polynomial-time computable functions $c, s: \N \to [0,1]$ if there exists a polynomial-time uniform family of quantum circuits $\left\{V_n \right\}_{n \in \N}$ such that for every $n$-bit input $x$, $V_n$ takes in proofs $y_1, \dots, y_i \subseteq \{0,1\}^{m(n)}$ for fixed polynomial $m$ and measures a fixed output qubit to decide to accept or reject, such that
    \begin{itemize}
        \item Completeness: $ x \in L_\yes \Rightarrow  \exists y_1 \forall y_2 \exists y_3 \dots Q_i y_i \text{ such that } \Pr\left[ V( y_1, \dots, y_i) \text{ accepts}\right] \geq c $,
        \item Soundness: $ x\in L_\no \Rightarrow  \forall y_1 \exists y_2  \forall y_3 \dots \overline{Q_i} y_i  \text{ such that } \Pr\left[ V(y_1, \dots, y_i) \text{ accepts}\right] \leq s $,
    \end{itemize}
   where $Q_i$ denotes $\exists$ if $i$ is odd and $\forall$ otherwise, $\overline{Q_i}$ denotes the complement of $Q_i$, and, for all $i$, $\abs{y_i} \leq p(n)$ for a fixed polynomially bounded function $p$. 
   When the completeness and soundness parameters $c, s$ are not specified, define
    \[
        \QCSigmai[i] \coloneqq \bigcup_{c-s \in \Omega\left(1/\poly(n)\right)} \QCSigmai[i]\class{(c,s)}.
    \]
\end{definition}
\begin{definition}[The quantum-classical polynomial hierarchy ($\QCPH$) {\cite{GSSSY22_qph}}]
The quantum-classical polynomial hierarchy is defined as 
\[
\QCPH \coloneqq \bigcup_{i=0}^\infty \QCSigmai.
\]
\end{definition}

Observe that $\QCSigmai[0] = \BQP$ and $\QCSigmai[1] = \QCMA$. Gharibian et al.\ \cite{GSSSY22_qph} proved that $\QCPH$ is contained in $\PTIME^{\PP^\PP}$, the second level of the counting hierarchy $\CH$. 

The definition of $\QCSigmai[i]$ to generically include $\QCSigmai[i]\class{(c,s)}$ for all $c-s\geq 1/\poly(n)$ is justified in part by the result of \cite{GSSSY22_qph} that for any such $c$ and $s$, we may reduce the error such that for any polynomially bounded function $r$, we have
$\QCSigmai[i]\class{(c,s)} = \QCSigmai[i]\class{(1-2^{-r},2^{-r})}$.

The unentangled quantum polynomial hierarchy $\QPH$ is defined similarly. The only difference is that the classical proofs are replaced by unentangled quantum proofs. 

\begin{definition}[$\QSigmai$ {\cite{GSSSY22_qph}}] \label{def:prodqsigmai}
    A promise problem $L=(L_\yes, L_\no)$ is in the $i$-th level of the unentangled quantum polynomial hierarchy $\QSigmai[i]\class{(c,s)}$ for polynomial-time computable functions $c, s: \N \to [0,1]$ if there exists a polynomial-time uniform family of quantum circuits $\left\{V_n \right\}_{n \in \N}$ such that for every $n$-bit input $x$, $V_n$ takes in quantum proofs $\rho_1, \dots, \rho_i$ and measures a fixed output qubit to decide to accept or reject, such that
    \begin{itemize}
        \item Completeness: $ x \in L_\yes \Rightarrow  \exists \rho_1 \forall \rho_2 \exists \rho_3 \dots Q_i \rho_i \text{ such that } \Pr\left[ V(\rho_1, \dots, \rho_i) \text{ accepts}\right] \geq c $,
        \item Soundness: $ x\in L_\no \Rightarrow  \forall \rho_1 \exists \rho_2  \forall \rho_3 \dots \overline{Q_i} \rho_i  \text{ such that } \Pr\left[ V(\rho_{1},\dots, \rho_i) \text{ accepts}\right] \leq s $,
    \end{itemize}
   where $Q_i$ denotes $\exists$ if $i$ is odd and $\forall$ otherwise, $\overline{Q_i}$ denotes the complement of $Q_i$, and, for all $i$, $\rho_i$ is a $p(n)$-qubit state for a fixed polynomially bounded function $p$. 
   When the completeness and soundness parameters $c, s$ are not specified, define
    \[
        \QSigmai[i] \coloneqq \bigcup_{c-s \in \Omega(1)} \QSigmai[i]\class{(c,s)}.
    \]
\end{definition}
\begin{definition}[$\QPH$ {\cite{GSSSY22_qph}}]
The unentangled quantum polynomial hierarchy is defined as 
\[
\QPH \coloneqq \bigcup_{i=0}^\infty \QSigmai.
\]
\end{definition}

\noindent Interestingly, $\QMAtwo \subseteq \QSigmai[3]$, since the verifier can simply ignore the second proof. 

Here, we let $\QSigmai[i] = \QSigmai[i]\class{(c,s)}$ for $c-s \geq \Omega(1)$, rather than $1/\poly(n)$, because we do not currently have an error reduction result for $\QPH$ similar to the one known for $\QCPH$ (although, \cite{agarwal2023quantum} recently made progress in this direction). 

\section{The Entangled Quantum Polynomial Hierarchy}\label{sec:prelim-entqph}

We formally define the entangled quantum polynomial hierarchy. 
The definition appears more technical than for $\QCPH$ and $\QPH$, but this is mostly just an issue of notation.

\begin{definition}[$i$-th level of the entangled quantum polynomial hierarchy ($\entQSigmai $)] \label{def:entqsigmai}
    A promise problem $L=(L_\yes, L_\no)$ is in $\entQSigmai(c,s)$ for polynomial-time computable functions $c, s: \N \to [0,1]$ if there exists a polynomial-time uniform family of quantum circuits $\left\{V_n \right\}_{n \in \N}$ such that for every $n$-bit input $x$, $V_n$ takes quantum proofs, measures a fixed output qubit to decide to accept or reject, and satisfies 
    \begin{itemize}
        \item Completeness: $ x \in L_\yes \Rightarrow  \exists \rho_1 \forall \rho_2 \exists \rho_3 \dots Q_i \rho_i \text{ such that } \Pr\left[ V(\rho_{i-1}, \rho_i) \text{ accepts}\right] \geq c $,
        \item Soundness: $ x\in L_\no \Rightarrow  \forall \rho_1 \exists \rho_2  \forall \rho_3 \dots \overline{Q_i} \rho_i  \text{ such that } \Pr\left[ V(\rho_{i-1}, \rho_i) \text{ accepts}\right] \leq s $,
    \end{itemize}
    where each $\rho_j$ is chosen from the set  
    \begin{equation*}
    \calA_j \coloneqq \begin{cases}
        \quad\left\{ \rho \in \mathbf{D}\left(\calX_1\otimes \calX_3 \otimes \dots \otimes \calX_{j}\right) \mid \text{if }j>1,\; \tr_{\calX_j}\left(\rho\right) = \rho_{j-2} \right\} & \text{if $j$ is odd}  \\
        \quad\left\{ \rho \in \mathbf{D}\left( \calX_2\otimes \calX_4 \otimes \dots \otimes \calX_{j}\right) \mid  \text{if }j>2,\; \tr_{\calX_i}\left(\rho\right) = \rho_{j-2} \right\} & \text{if $j$ is even}
    \end{cases} .
    \end{equation*}
    Here, $Q_i$ denotes $\exists$ if $i$ is odd and $\forall$ otherwise, and $\overline{Q_i}$ denotes the complement of $Q_i$. For all $i$, the corresponding Hilbert space $\calX_i$ is a space of at most $p(n)$ qubits for a fixed polynomial $p$.
    When the completeness/soundness parameters are not specified, define
    \[
        \entQSigmai[i] \coloneqq \bigcup_{c-s \in \Omega\left(1/\poly(n)\right)} \entQSigmai[i](c,s).
    \]
\end{definition}

\begin{definition}[The entangled quantum polynomial hierarchy ($\QEPH$)]\label{def:entqph}
    The entangled quantum polynomial hierarchy is defined as 
    \[\QEPH = \bigcup_{i = 0}^{\infty} \entQSigmai[i].\]
\end{definition}

When introducing a complexity class, perhaps the first question one should ask is whether or not the choice of completeness and soundness parameters actually matter. 
In \cite[Theorem 2.6]{GSSSY22_qph}, it was shown that $\QCPH$ is robust to the choice of error parameters, but no such result is known for $\QPH$.
In \cref{sec:entQPHcollapseInPSPACE}, we show that the choice of parameters does not matter for any level of $\QEPH$, i.e., for $c, s$ such that $c-s \geq 1/\poly(n)$, $\entQSigmai(c,s) = \entQSigmai(\frac{2}{3}, \frac{1}{3})$ for all $i \in \N$ (see \cref{thm:entqsigmaerror}).

Let us also make several remarks on our definition. 
As for $\PH$, the indices $i$ in the definition of $\QEPH$ are constants, independent of a problem's input size, and, as one should expect, $\BQP = \entQSigmai[0]$ and $\QMA=\entQSigmai[1]$. 
One can also define $\entQPii[i] \coloneqq \class{\overline{\Sigma}_i^{qe}}$ and $\entQDeltai[i] \coloneqq \entQSigmai \cap \entQPii$. 
The players also have no incentive to entangle their moves with their opponent because $\entQSigmai$ can be modeled as a zero-sum game. 
Therefore, we may assume the even and odd indexed states are unentangled.

Informally, $\entQSigmai$ can be thought of as the following game, where we assume $i$ is even to simplify the exposition.
Alice has (possibly entangled) quantum registers $(\qreg{A_1}, \ldots, \qreg{A_{i/2}})$,
and Bob has (possibly entangled) quantum registers $(\qreg{B_1}, \ldots, \qreg{B_{i/2}})$,
where each register is a number of qubits that is polynomial in the input size. 
The game commences as follows. In the first round, Alice reveals the state $\rho_1$ of $\qreg{A_1}$, and then Bob reveals the state $\sigma_1$ of $\qreg{B_1}$. 
In the second round, Alice reveals the state $\rho_2$ of $(\qreg{A_1}, \qreg{A_2})$, and Bob reveals the state $\sigma_2$ of $(\qreg{B_1}, \qreg{B}_2)$. 
To ensure Alice and Bob do not change their ``moves'' from previous rounds, we demand that $\tr_{\qreg{A_2}}(\rho_2) = \rho_1$ and $\tr_{\qreg{B_2}}(\sigma_2) = \sigma_1$. That is, Alice and Bob cannot modify the state of subsystems that have been revealed in previous rounds. 
In general, for the $i$-th round, it must be that 
$\tr_{\qreg{A_i}}(\rho_i) = \rho_{i-1}$ and 
$\tr_{\qreg{A_i}}(\sigma_i) = \sigma_{i-1}$.
The game continues like this until the global states of $(\qreg{A_1}, \ldots, \qreg{A_{i/2}})$  and $(\qreg{B_1}, \ldots, \qreg{B_{i/2}})$ are known to both players and the referee. 

At this point, the referee must accept or reject. 
The referee's action is determined by a polynomial-time quantum circuit and a single-qubit measurement. This action can be equivalently expressed as a two-outcome quantum measurement $\{R, I-R\}$, where the first observable corresponds to accepting. Then, the probability the referee accepts is equal to $\tr\left(R \left(\rho_{i/2} \otimes \sigma_{i/2} \right) \right)$. 
We emphasize that we do not intend to actually write the observable $R$ corresponding to some verification circuit $V$. Rather, the observable $R$ is a convenient way to express the action of the referee. 

Alice's goal is to maximize the acceptance probability, and Bob's goal is to minimize the acceptance probability. 
Therefore, given an instance of an $\entQSigmai[i]$ problem with corresponding observable $R$, 
we can express the acceptance probability achieved by both players playing optimal strategies as
\begin{equation}\label{eqn:minmaxdef}
    \upsilon = \max_{\rho_1 \in \calA_1} \min_{\sigma_1 \in \calB_1} \ldots \max_{\rho_{i/2} \in \calA_{i/2}} \min_{\sigma_1 \in \calB_{i/2}} \tr\left(R \left(\rho_{i/2} \otimes \sigma_{i/2}\right) \right),   
\end{equation}
where $\calA_i$ and $\calB_i$ are defined as in \cref{def:entqsigmai}, and each alternating max/min operator corresponds to an alternation of quantifiers in \cref{def:entqsigmai}.
In this work, we intend to use \cref{eqn:minmaxdef} as a tool for proving the equality of one game/problem instance to another.

Finally, as an application of \cref{eqn:minmaxdef}, we observe that the second levels of both $\QEPH$ and $\QPH$ are equal to $\QRGone$, which is known to be contained in $\PSPACE$.

\begin{proposition}[Extension of {\cite[Corollary 1.9]{GSSSY22_qph}}]\label{prop:second-levels-equal}
    \[
    \entQSigmai[2] = \entQPii[2] = \QSigmai[2] = \QPii[2] = \QRGone \subseteq \PSPACE.
    \]
\end{proposition}
\begin{proof}
    In \cite{GSSSY22_qph}, it was observed that $\QSigmai[2] = \QRGone$.
    Here, we use the same reasoning to conclude
    \[
        \entQSigmai[2] = \entQPii[2] = \QSigmai[2] = \QPii[2] = \QRGone .
    \]
    The equivalence is clear given that the value of a $\QRGone$ protocol is described by an expression identical to \cref{eqn:minmaxdef} when $i=2$, which corresponds to $\entQSigmai[2]$ (see \cite{jain2009parallel} for a formal definition of $\QRGone$).
    Then, note that $\QRGone$ is closed under complement, by a min-max theorem, implying $\entQSigmai[2]=\entQPii[2]$.
    Second, because entanglement is not a concern until one of the players makes multiple moves, the second levels of the entangled and unentangled hierarchies are equal (similarly, the first levels are equal to each other, as are the zeroeth levels).
    Finally, the containment of $\QRGone$ in $\PSPACE$ is due to \cite[Proposition 4]{jain2009parallel}.
\end{proof}

The fact that $\entQSigmai[2] = \entQPii[2] = \QSigmai[2] = \QPii[2]$ is somewhat striking, since such an equality in the classical setting would imply a collapse of $\PH$ \cite[Theorem 5.6]{arora2009computational}.\footnote{This phenomenon of the second levels being equal is also true for $\mathsf{TFPH}$, the hierarchy generalizing the class $\mathsf{TFNP}$ \cite{kleinberg_et_al:LIPIcs.ITCS.2021.44}.} 

\begin{remark}[Public vs. private quantum proofs]\label{remark:public}
While the quantum polynomial hierarchies are well-defined, some may object that the classes are unphysical because the provers have full knowledge of each other's density matrices, 
even though the verifier only receives a single copy of each proof. 
The quantum no-cloning theorem also begs the question of how exactly the information is communicated between the provers. 
This is not an issue for $\PH$ because it's trivial to learn classical proofs given a single copy, and, for $\QRG$, this is not an issue because communication is private. 
In every quantum complexity class that we are aware of, the provers are ``all-powerful'' yet still bound by the laws of quantum mechanics.

Despite being unphysical, we are content with the definition for two reasons. First, it is a well-defined, useful theoretical tool for studying quantum information. 
Second, in the case of $\QEPH$, we show that it collapses to $\entQSigmai[2]$, where it is known, by a min-max theorem, that public vs. private communication is irrelevant. So, despite starting with an unphysical definition, we show equivalence with a class that adheres entirely to the laws of quantum mechanics. 
\end{remark}

\section{The Entangled Quantum Polynomial Hierarchy Collapses}\label{sec:entQPHcollapseInPSPACE}

We prove several results about the entangled quantum polynomial hierarchy. 
Specifically, we prove that $\QEPH$ collapses to its second level, is equal to $\QRGone$, and that every level of $\QEPH$ is robust to the choice of completeness and soundness parameters (i.e., for $c, s$ such that $c-s \geq 1/\poly(n)$, $\entQSigmai(c,s) = \entQSigmai(\frac{2}{3}, \frac{1}{3})$ for all $i \in \N$).
We begin by proving that the hierarchy collapses.

\begin{lemma}\label{lem:ent-qph-collapse}
 For all constants $i \geq 2$, $\entQSigmai[2] = \entQSigmai[i]$.
\end{lemma}
\begin{proof}
Note that for all $i$, $\entQSigmai[i-1]$ is trivially contained in $\entQSigmai[i]$.
We will show that for all $i > 2$, $\entQSigmai[i] \subseteq \entQSigmai[i-1]$ by an induction argument, beginning with $\entQSigmai[3] \subseteq \entQSigmai[2]$.

Recall from \cref{eqn:minmaxdef} in \cref{sec:prelim-entqph} that the value of an $\entQSigmai[3]$ protocol is equal to 
\[
    \hat{\upsilon} = \max_{\rho_1 \in \mathbf{D}(\calX_1)}~~\min_{\sigma_1 \in \mathbf{D}(\calY_1)}~~\max_{\rho_{2} \in \calA}~~\tr\left( R \left(\rho_2 \otimes \sigma_1\right)\right) ,
\]
where $R$ is the observable corresponding to the verifier accepting,
$\calX_1$, $\calY_1$, and $\calX_2$ are the Hilbert spaces containing the three proofs, and
$\calA = \{ \rho \in \mathbf{D}(\calX_1 \otimes \calX_2) \mid \tr_{\qreg{A_2}}(\rho) = \rho_1 \}$, which enforces that Alice's second proof is consistent with her first.

For any choice of $\rho_1\in\mathbf{D}(\calX_1)$, define
\begin{equation*}
    \upsilon(\rho_1) = \min_{\sigma_1 \in \mathbf{D}(\calY_1)}~~\max_{\rho_{2} \in \calA}~~\tr\left( R \left(\rho_2 \otimes \sigma_1\right)\right) ,
\end{equation*}
so that $\hat{\upsilon} = \max_{\rho_1 \in \mathbf{D}(\calX_1)} \upsilon(\rho_1)$.
Consider that $\mathbf{D}(\calY_1)$ and $\calA$ are compact and convex by \cref{fact:density-matrices-compact-convex,fact:partial-trace-compact-convex}.
Additionally, the function $\tr\left( R \left(\rho_2 \otimes \sigma_1\right)\right)$
is a composition of bilinear functions and so itself is bilinear in $\sigma_1$ and $\rho_1$.
Therefore, by \cref{fact:min-max}, a min-max theorem applies and
\[
    \upsilon(\rho_1) =
    \max_{\rho_{2} \in \calA}~~\min_{\sigma_1 \in \mathbf{D}(\calY_1)}~~\tr\left( R \left(\rho_2 \otimes \sigma_1\right)\right)
    =
    \min_{\sigma_1 \in \mathbf{D}(\calY_1)}~~\max_{\rho_{2} \in \calA}~~\tr\left( R \left(\rho_2 \otimes \sigma_1\right)\right) ,
\]
changing the optimization problem without changing the value.

Substituting this back into $\hat{\upsilon}$, we find
\begin{alignat}{3}
    \hat{\upsilon} &=& \underset{\rho_1 \in \mathbf{D}(\calX_1)}&{\max}~~\upsilon(\rho_1) \nonumber \\
    &=& \underset{\rho_1 \in \mathbf{D}(\calX_1)}&{\max}~~\max_{\rho_{2} \in \calA}~~\underset{\sigma_1 \in \mathbf{D}(\calY_1)}{\min}~~&\tr&\left( R \left(\rho_2 \otimes \sigma_1\right)\right) \nonumber \\
    &=& \underset{\rho_2 \in \mathbf{D}(\calX_1 \otimes \calX_2)}&{\max}~~\min_{\sigma_1 \in \mathbf{D}(\calY_1)} &\tr&\left( R \left(\rho_2 \otimes \sigma_1\right)\right) , \label{eqn:qsigma3to2_minmax}
\end{alignat}
where the final equality is clear given the definition of $\calA$.

We observe that \cref{eqn:qsigma3to2_minmax} matches the characterization of an $\entQSigmai[2]$ protocol given in \cref{eqn:minmaxdef}.
Therefore, we have shown the value $\hat{\upsilon}$ of an arbitrary $\entQSigmai[3]$ protocol is equivalent to the value of an $\entQSigmai[2]$ protocol.
Given an instance of an $\entQSigmai[3]$ problem verified by some polynomial-time uniform circuit $V$---corresponding to the observable $R$ above---whether $V$ is satisfiable by an $\entQSigmai[3]$ protocol is equivalent to whether $V$ is satisfiable by an $\entQSigmai[2]$ protocol,
i.e.\, $\entQSigmai[3] \subseteq \entQSigmai[2]$ and indeed they are equal. 

By way of induction, assume $\entQSigmai[2] = \entQSigmai[i]$ for some constant $i > 2$.
By the same min-max argument as just before, we may show the equivalence of the value of any $\entQSigmai[i+1]$ protocol to the value of a $\entQSigmai[i]$ protocol, thus showing the equivalence of the classes.
Therefore, the hierarchy $\QEPH$ collapses to $\entQSigmai[2]$.
\end{proof}

The equality between $\QEPH$ and $\QRGone$ is a straightforward consequence of the collapse lemma.

\begin{theorem}\label{thm:qrg-one-equals-pspace}
    $\QRGk[1] = \QEPH = \entQSigmai[2]$.
\end{theorem}
\begin{proof}
    Combining the results $\QRGk[1] = \entQSigmai[2]$ from \cref{prop:second-levels-equal} and $\entQSigmai[2]=\QEPH$ from \cref{lem:ent-qph-collapse} proves the equality.
\end{proof}

Next, we note that our collapse theorem can be strengthened to $\entQSigmai[i]=\entQSigmai[2]$ for any polynomially bounded $i$, rather than just constant. 
Like classical $\PH$, we define $\QEPH$ as the union of $\entQSigmai[i]$ for any constant $i$. This is a natural way of defining $\PH$ as it is key to proving that if $\PTIME=\NP$, then $\PH$ collapses. 
However, in contrast to collapse techniques for classical $\PH$, our reduction of $\entQSigmai[i]$ to $\entQSigmai[2]$ does not increase the problem size. In our proof of \cref{lem:ent-qph-collapse}, the $\entQSigmai[2]$ problem in \cref{eqn:qsigma3to2_minmax} optimizes over the same quantity as the original $\entQSigmai[i]$ problem. Therefore, our proof applies even to a super-constant number of rounds. The reduction is valid up to a polynomial number of rounds, after which the concatenation of the proof registers would lead to a proof too large for the polynomial-time verifier to accept.

\begin{corollary}\label{corr:entQSigmaPoly=2}
    $\entQSigmai[i] = \entQSigmai[2]$ for any polynomially-bounded $i$.
\end{corollary}

Finally, our results also prove that $\entQSigmai$ is robust to the choice of error parameters.

\begin{theorem}\label{thm:entqsigmaerror}
    For any choice of $c,s$ such that $c-s\geq 1/\poly(n)$, it holds that $\entQSigmai[i]\left(c,s\right) = \entQSigmai[i]\left(\tfrac{2}{3},\tfrac{1}{3}\right)$.
\end{theorem}
\begin{proof}
    The reverse containment is trivial, so we focus on proving the forward direction, reducing $\entQSigmai[i]\left(c,s\right)$ to $ \entQSigmai[i]\left(\tfrac{2}{3},\tfrac{1}{3}\right)$.
    Again appealing to the fact that our proof of \cref{lem:ent-qph-collapse} shows that an $\entQSigmai[3]$ problem is equivalent to an $\entQSigmai[2]$ problem with the same game value, we observe that our proof implies $\entQSigmai[2](c,s)=\entQSigmai[i](c,s)$.
    Then, because the equality of $\QRGone$ and $\entQSigmai[2]$
    (\cref{thm:qrg-one-equals-pspace}) is also based on the optimization definition from \cref{eqn:minmaxdef}, the acceptance probability remains preserved and $\entQSigmai[2](c,s)=\QRGone(c,s)$. We may then appeal to the result of \cite{gutoski2005qrg} that a parallel repetition theorem holds for $\QRGone$, so that $\QRGone(c,s)=\QRGone\left(\tfrac{2}{3},\tfrac{1}{3}\right)$.
    By the same reasoning as a moment ago, this last class equals $\entQSigmai[2]\left(\tfrac{2}{3},\tfrac{1}{3}\right)$.
    Contracting this sequence of equalities, we conclude that $\entQSigmai[i]\left(\tfrac{2}{3},\tfrac{1}{3}\right)$ equals our original class $\entQSigmai[i](c,s)$.
\end{proof}

\section{\texorpdfstring{$\PH$ and $\QCPH$ Are Contained in $\QPH$}{PH and QCPH Are Contained in prodQPH}}\label{sec:direct-proof-ph-prodqph}

We prove that $\QCPH \subseteq \QPH$. 
While this result is what one might expect, proving this containment was left as an option question by Gharibian et al.\ \cite{GSSSY22_qph}.
It is trivial to see that $\PH \subseteq \QCPH$, and, combining these two containments, we have $\PH \subseteq \QCPH \subseteq \QPH$, establishing that quantifying over unentangled quantum proofs is at least as powerful as quantifying over classical proofs. 

The central challenge in proving that $\QCPH \subseteq \QPH$ is that the proofs in $\QPH$ are allowed to be quantum states, which, upon measurement, give rise to a distribution over classical strings. 
A flawed idea is to simply measure the quantum proofs to get classical proofs, and then run the $\QCPH$ verification protocol with no modifications. 
Suppose, however, that Alice has a winning strategy in the $\QCPH$ protocol, so she always has a winning response to any classical proof that Bob sends. 
When simulating this in $\QPH$, Bob can instead send a \emph{quantum state}---a superposition over many classical proofs---preventing Alice from sending an optimal response. In particular, Alice may not know which response to send, since she doesn't know which classical proof the verifier will observe upon measurement.  

We prevent this potential cheating by requiring each player to send multiple copies of each of their proofs. We prove that this is enough to force both players to send classical strings with high probability. 

\begin{theorem}\label{thm:direct-proof-ph-prodqph}
    $\PH \subseteq \QCPH \subseteq \QPH$.
\end{theorem}

The exact error parameters for \cref{thm:direct-proof-ph-prodqph} are stated in \cref{eqn:QCinQPH} below.
In particular, the reduction is only capable of producing a $\QPH$ instance with a constant promise gap.
However, the containment does hold for any $\QCPH$ instance with at least an inverse-polynomial promise gap, due to known error reduction for $\QCPH$ \cite{GSSSY22_qph}.

\begin{proof}
    Consider any level $\QCSigmai$ of $\QCPH$. 
    We show that for any integer $k\geq 1$, 
    \begin{equation}\label{eqn:QCinQPH}
        \QCSigmai(c,s) \subseteq \QSigmai[2ki]\left(c\left(1-2^{-k}\right),s + 2^{-k}\left(1 - s\right)\right).
    \end{equation} 

    We simulate any $\QCSigmai$ protocol in $\QSigmai[2ki]$ as follows. 
    After the first $2k$ turns, the verifier has $k$ proofs from Alice and $k$ proofs from Bob, and the verifier discards all $k$ proofs from Bob.
    For the next $2k$ turns, the verifier repeats this process, except they keep Bob's proofs rather than Alice's, which we denote by $\sigma_{1,1},\dots, \sigma_{1,k}$.
    This is repeated $i$ times in total until all $2ki$ turns are over. At the end of the game, the verifier has kept the following $ki$ proofs:
    \[
    \rho_{1,1},\dots, \rho_{1,k}, \sigma_{1,1}, \dots, \sigma_{1,k}, \rho_{2,1},\dots, \rho_{2,k}, \dots .
    \]
    
    For each chunk of $k$ proofs, the verifier measures each quantum state in the standard basis to get $k$ classical strings. If all $k$ classical strings are equal, we say that the player passed the check, and failed otherwise.  
    If a player fails any check, then the other player is declared the winner.
    If both players pass all checks, then the verifier keeps one copy of each classical proof from each chunk and runs the $\QCPH$ verification procedure to determine the winner. 

    Let $A = (A_\yes, A_\no)$ be a promise problem in $\QCSigmai(c,s)$, and let $x$ be some fixed input. 
    If $x \in A_\yes$, then Alice has no incentive to cheat and so we refer to her as the honest prover, while if $x \in A_\no$ then we consider Bob the honest prover.
    We will define a strategy for the honest prover and show that no matter the strategy of the dishonest prover, the honest prover will win high probability. 
    In particular, the honest prover's strategy will be to always send classical proofs, and when replying to a dishonest provers proof $\rho = \sum_j p_j \ket{j}\!\!\bra{j}$, the honest prover will respond as if only the string $\hat{\text{\j}}$ with the maximum probability $p_{\hat{\text{\j}}}$ was sent (we arbitrarily choose to break ties by lexicographic order).

    If the dishonest prover fails any check, they lose, so we assume now that the dishonest prover passes every check.
    Then, since both provers pass every check, the verifier has the $i$ classical proofs $y_1, \dots, y_i$, where the proofs with odd indices are from Alice and the others are from Bob. 
    In one case, suppose that each of the dishonest prover's moves turns out to be as the honest prover expected. 
    Then the situation is identical to the original $\QCPH$ instance, and so the honest prover wins with the probability of the original protocol. 
    
    In the second case, at least one chunk of $k$ proofs (sampled independently from $k$ distributions) 
    are equal to each other but not to the proof $\hat{\text{\j}}$ expected by the honest prover. 
    Any string besides $\hat{\text{\j}}$ has $p_j \leq 1/2$,
    so the probability of this case occurring, with all $k$ samples matching, is at most $2^{-k}$. 
    
    Therefore, in the $\QPH$ protocol, if $x \in A_\yes$, Alice wins with probability at least $c\left(1 -2^{-k}\right)$. If $x \in A_\no$, then Bob wins with probability at least $(1-s)\left(1-2^{-k}\right)$, so Alice wins with probability at most 
    \[
    1 - (1-s)\left(1-2^{-k}\right) = s + 2^{-k}(1 - s).
    \]
    To summarize, the dishonest prover is unable to affect the outcome of the game with more than a small probability. 
    We conclude that $\QCSigmai \subseteq \QSigmai[2ki]$, and therefore $\QCPH\subseteq \QPH$. 
\end{proof}

\section{Distribution Hierarchies}\label{sec:mixed}

We introduce another generalization of the polynomial hierarchy where the provers send probability distributions over bit strings. 
This gives rise to two new hierarchies: the distributional polynomial hierarchy $\distributionPH$ and its quantum analogue $\distributionQCPH$, which is the same as $\distributionPH$ but with a quantum verifier.
We will focus primarily on $\distributionPH$ since the techniques used to analyze $\distributionPH$ will work for $\distributionQCPH$ as well.

$\mixedPH$ is similar to all of the hierarchies studied in this work.
In $\mixedPH$, the distributions are public (the provers have full knowledge of the distributions that have been sent), but none of the distributions are sampled until every distribution has been sent.
One can think of this as a non-interactive game, where the players use public, mixed strategies. 
Importantly, the distributions are not correlated across rounds. 

While $\mixedPH$ is a classical complexity class, our motivation for studying it is to further understand the quantum polynomial hierarchies. 
In particular, $\mixedPH$ involves proofs that are classical mixtures of bit strings. 
This complements $\pureQPH$, where the proofs are quantum superpositions of bit strings, and $\QPH$, where the proofs are both (classical mixtures of quantum superpositions).
Does the computational power of the polynomial hierarchy increase when the proofs only involve classical probability distributions? Or does the increased computational power come only from the quantum superposition allowed in $\QPH$ and $\pureQPH$?
In this section, we resolve these questions. 

\begin{theorem}\label{thm:mixedph=ph}
    $\mixedPH = \PH$.
\end{theorem}

That is, if the proofs are distributions over classical proofs, $\PH$ does not increase in power.
The proof of \cref{thm:mixedph=ph} relies on a technical lemma that says the distributions sent in $\mixedPH$ can be sparse and uniform. This lemma generalizes a result due to Lipton and Young \cite{lipton1994simple} and Althöfer \cite{ALTHOFER1994339}. 

In the remainder of this section, we will formally define $\distributionPH$, prove the technical lemma, and prove \cref{thm:mixedph=ph}. 
Finally, we will discuss $\distributionQCPH$ (the same as $\distributionPH$ but with a quantum verifier) and the power of classical versus quantum proofs.

We begin by formally defining $\mixedPH$.
Let $\calD_m$ denote the set of all probability distributions over $\{0,1\}^m$.
For a computation $M$ which takes length-$m$ strings as input and a distribution $\rho\in\calD_m$, let $M(\rho)$ implicitly refer to $M(y)$ for $y\sim \rho$.

\begin{definition}[$i$-th level of the distribution polynomial hierarchy ($\mixedSigmai$)]
    A promise problem $L =(L_{\text{yes}},L_{\text{no}})$ is in $\mixedSigmai(c,s)$ for polynomial-time computable functions $c,s : \N \to [0,1]$ if there exists a classical polynomial-time randomized Turing Machine $M$ such that
    \begin{itemize}
        \item Completeness: $x\in L_{\text{yes}} \Rightarrow \exists \rho_1 \forall \rho_2 \exists \rho_3 \dots Q_i \rho_i$ such that $\Pr\left[M(\rho_1,\dots,\rho_i) = 1\right] \geq c$,
        \item Soundness: $x\in L_{\text{no}} \Rightarrow \forall \rho_1 \exists \rho_2 \forall \rho_3 \dots \overline{Q_i} \rho_i$ such that $\Pr\left[M(\rho_1,\dots,\rho_i) = 1\right] \leq s$,
    \end{itemize}
    where each $\rho_k$ is a distribution in $\calD_m$ for some polynomially-bounded $m$, and each $\rho_k$ is independent.
    $Q_i$ is $\exists$ if $i$ is odd and $\forall$ otherwise, and $\overline{Q_i}$ is the complement of $Q_i$.
    When the completeness/soundness parameters are not specified, define
    \[
        \mixedSigmai[i] \coloneqq \bigcup_{c-s\in\Omega(1)} \mixedSigmai[i](c,s).
    \]
\end{definition}

\begin{definition}[The distribution polynomial hierarchy ($\distributionPH$)]\label{def:mixedPH}
    The distribution polynomial hierarchy is defined as 
    \[\mixedPH = \bigcup_{i = 0}^{\infty} \mixedSigmai[i].\]
\end{definition}

We make a few comments on our definition of $\distributionPH$.
If we defined $\mixedPH$ without the bounded-error condition (i.e., no error probability), then it would be equal to $\PH$. We will also generally leave the input $x$ implicit.
Finally, if one prefers, they can equivalently think of the provers sending quantum mixed states that are immediately measured in the computational basis (instead of probability distributions that are immediately sampled). This is why we choose to denote the probability distributions as $\rho_i$ in our definition. 

As we discussed in \cref{sec:prelim-entqph} for $\QEPH$, one can think of $\mixedPH$ as a game, where two competing provers take turns sending distributions over bit strings to a verifier. Then the verifier $M$ draws one sample from each distribution and runs a polynomial-time randomized algorithm to determine a winner.
Additionally, just like with $\QEPH$, we can express the acceptance probability of the verifier as the following optimization problem:
\[
\Pr[M \text{ accepts}] = \max_{\rho_1 \in \calD_m} \min_{\rho_2 \in \calD_m} \ldots \operatorname*{Q^i}_{\rho_i \in \calD_m} \E[M(\rho_1, \ldots, \rho_i)],
\]
where
$\operatorname*{Q^i}$ denotes $\max$ if $i$ is odd and $\min$ otherwise. 
The expectation is over the randomness in the distributions $\rho_1, \ldots, \rho_i$. Note that since $M(\rho_1, \ldots, \rho_i)$ is a Bernoulli random variable, $\E[M(\rho_1, \ldots, \rho_i)] = \Pr[M(\rho_1, \ldots, \rho_i) = 1]$.

The distributions sent in $\mixedPH$ are over $\{0,1\}^m$ for some polynomially-bounded $m$, so, in general, the support can be exponentially large in $m$. 
We will prove a technical lemma that says the provers can send \emph{uniform} distributions over $\poly(m)$ bit strings without changing the outcome of the game too much. 

\begin{lemma}\label{lem:lipton}
    For any constant $k \in \N$ and any classical randomized Turing Machine $M$ accepting $k$ length-$m$ inputs, if
    \[
        \max_{\rho_1 \in\calD_m} \min_{\rho_2 \in\calD_m} \max_{\rho_3 \in\calD_m} \dots \operatorname*{Q^k}_{\rho_k \in\calD_m} \Pr\left[M(\rho_1,\dots,\rho_k) = 1\right] = v,
    \]
    then
    for any constant $\epsilon>0$, 
    \[
        \max_{\rho_1\in U_{t_k}} \min_{\rho_2\in U_{t_{k-1}}} \max_{\rho_3 \in U_{t_{k-2}}} \dots \operatorname*{Q^k}_{\rho_k \in U_{t_{1}}} \Pr\left[M(\rho_1,\dots,\rho_k)=1\right] \in [v - k\epsilon, v+k\epsilon] ,
    \]
    where $t_i \coloneqq \lceil m^{2i}/2\epsilon^2 \rceil$, 
    $U_t$ denotes the set of uniform distributions over multi-sets of size at most $t$ of strings in $\{0,1\}^m$,
    and $\operatorname*{Q^k}$ denotes $\max$ if $k$ is odd and $\min$ otherwise.
    The complement of this result also holds (i.e., when the sequence starts with $\min$ instead of $\max$).
\end{lemma}

\begin{proof}
We will prove the claim by induction. The base case $k=2$ is precisely \cite[Theorem 2]{lipton1994simple} (see also \cite{ALTHOFER1994339}). Our contribution is to generalize their result to larger $k$.

    By way of induction, suppose the claim holds for $k-1$, and consider an instance with $k$ rounds:
    \[
        v \coloneqq \max_{\rho_1 \in\calD_m} \min_{\rho_2 \in\calD_m} \max_{\rho_3 \in\calD_m} \dots \operatorname*{Q^k}_{\rho_{k} \in\calD_m} \Pr\left[M(\rho_1,\dots,\rho_k) = 1\right].
    \]
   Since the complement of this result (where a $\min$ is first instead of a $\max$) follows in the same way, we omit the details.
    
    Fix $\rho_1$ to a distribution that maximizes the acceptance probability (and think of $\rho_1$ as hardcoded into the input). 
    Consider the inner $k-1$ distributions $\rho_2, \ldots, \rho_k$. By the inductive hypothesis, we can simplify these distributions to 
    \[
   \min_{\rho_2 \in U_{t_{k-1}}} \max_{\rho_3 \in U_{t_{k-2}}} \dots \operatorname*{Q^k}_{\rho_{k} \in U_{t_1}} \Pr\left[M(\rho_1,\dots,\rho_k) = 1\right], 
    \]
    while only changing the acceptance probability $v$ by $\pm (k-1)\eps$. 
    In particular, we have that 
    \[
        v^\prime \coloneqq \max_{\rho_1 \in\calD_m} \min_{\rho_2 \in U_{t_{k-1}}} \max_{\rho_3 \in U_{t_{k-2}}} \dots \operatorname*{Q^k}_{\rho_{k} \in U_{t_1}} \Pr\left[M(\rho_1,\dots,\rho_k) = 1\right] \in [v - (k-1) \eps, v + (k-1)\eps]. 
    \]
    We want to show that we can simplify the first distribution $\rho_1$ in a similar fashion. Specifically, we want to show 
    \[
     v^{\prime\prime} \coloneqq  \max_{\rho_1 \in U_{t_k}} \min_{\rho_2 \in U_{t_{k-1}}} \max_{\rho_3 \in U_{t_{k-2}}} \dots \operatorname*{Q^k}_{\rho_{k} \in U_{t_1}} \Pr\left[M(\rho_1,\dots,\rho_k) = 1\right] \in [v - k \eps, v + k\eps]. 
    \]
    
    Observe that choosing $\rho_1$ from $U_{t_k}$ instead of $\calD_m$ can only hurt the maximizing player. That is, the probability that $M$ accepts can only \emph{decrease}, so $v^{\prime\prime} \leq v^\prime + \eps \leq v + k\eps$ is trivial. 
    All that remains is to show that $v^{\prime\prime} \geq v - k\eps$. 
    To prove this, it suffices to show that $v^{\prime\prime} \geq v^\prime - \eps$.
    
    Let $\rho_1^* \in \calD_m$ be a distribution that maximizes the acceptance probability of $M$. 
    Form a multi-set $S$ by drawing $t_k$ independent samples from $\rho_1^*$.
    Consider a string $y \in S$. This gives rise to a random variable on the interval $[0,1]$: 
    \[
    \E_{\rho_2, \ldots, \rho_k}[M(y, \rho_2, \ldots, \rho_k)], 
    \]
    where we are taking the expectation over optimal choices of $\rho_2, \ldots, \rho_k$.
    In expectation over $\rho_1^*$, we have 
    \[
    \E_{y \sim \rho_1^*}\left[ 
    \E_{\rho_2, \ldots, \rho_k}[M(y, \rho_2, \ldots, \rho_k]\right] = v^\prime.
    \]
    Therefore, by Hoeffding's inequality (\cref{fact:hoeffding}), 
    \[
    \Pr\left[ \frac{1}{\abs{S}} \sum_{y \in S} \E_{\rho_2, \ldots, \rho_k}[M(y, \rho_2, \ldots, \rho_k)] \leq v^\prime - \eps \right] \leq \exp\left(-2 t_k \eps^2\right).
    \]
    
    To complete the proof, we must count the number of sequences of distributions the minimizing player can send. 
    The minimizing player sends at most $k/2$ of the distributions $\rho_2, \ldots, \rho_k$, each of which is a uniform distribution over at most $t_{k-1}$-sized subsets of $\{0,1\}^m$. 
    Therefore, in total, there are at most
    \begin{gather*}
        \left(\sum_{i=1}^{t_{k-1}} \binom{2^m}{i} \right)^{k/2} \leq \left(\sum_{i=1}^{t_{k-1}} 2^{im} \right)^{k/2} \leq \left(t_{k-1} 2^{mt_{k-1}}\right)^{k/2} = t^{k/2}_{k-1} 2^{kmt_{k-1}/2}
    \end{gather*}
    possible sequences.
    We want to choose $t_k$ so that 
    \begin{equation}\label{eq:first}
    \exp\left(-2 t_k \eps^2\right) < \frac{1}{t^{k/2}_{k-1} 2^{kmt_{k-1}/2}},
    \end{equation}
    which would imply that strictly less than $1$ of the minimizing player's sequences of distributions can decrease $v^\prime$ by more than $\eps$. 
    Or, more directly, it would imply that there are no sequences the minimizing player can send to decrease $v^\prime$ by more than $\eps$. 
    We will show that choosing $t_k = m^{2k}/2\eps^2$ suffices. 
    By plugging in the definitions of $t_k$ and $t_{k-1}$, \cref{eq:first} becomes 
    \begin{equation}\label{eq:second}
    \exp\left(-m^{2k}\right) < \frac{\eps^k}{ m^{k(k-1)}} 2^{\frac{k}{2}-\frac{km^{2k -1}}{4 \eps^2}} \iff 
    \exp\left(-m^{2k}\right)  \frac{ m^{k(k-1)}}{\eps^k} 2^{\frac{km^{2k -1}}{4 \eps^2}-\frac{k}{2}} < 1. 
    \end{equation}
    We show that the inequality in \cref{eq:second} holds, which proves that our setting of $t_k$ is correct. 
    \begin{align*}
    \exp\left(-m^{2k}\right)  \frac{ m^{k(k-1)}}{\eps^k} 2^{\frac{km^{2k -1}}{4 \eps^2}-\frac{k}{2}} 
    &< \exp\left(-m^{2k}\right) m^{k^2} 2^{\frac{km^{2k -1}}{4 \eps^2}-\frac{k}{2}}  \\
    &< m^{k^2} 2^{\frac{km^{2k -1}}{4 \eps^2}-\frac{k}{2} - m^{2k}}  \\
    &= m^{k^2} 2^{m^{2k -1}\left(\frac{k}{4 \eps^2}-\frac{k}{2m^{2k-1}} - m\right)}  \\
    &< m^{k^2} 2^{-m^{2k -1}}  \\
    &< 1. 
    \end{align*}
    The first inequality holds because $m^{k} > \eps^{-k}$ for constant $\eps > 0$. The second-to-last inequality holds because $\left(\frac{k}{4 \eps^2}-\frac{k}{2m^{2k-1}} - m\right) < -1$ for constant $\eps > 0$.

   We conclude that $v^{\prime \prime} \geq v^\prime - \eps \geq v - k\eps$, which completes the proof. 
\end{proof}

We can now prove that $\mixedPH = \PH$. 

\begin{proof}[Proof of \cref{thm:mixedph=ph}]
    $\PH \subseteq \mixedPH$ follows from the proof that $\PH\subseteq \QPH$.
    This only achieves containment in $\mixedPH$ with constant promise gap, and it puts the $k$-th level of $\PH$ in some higher level of $\mixedPH$ (see \cref{thm:direct-proof-ph-prodqph} for more detail).
    
    To show $\mixedPH \subseteq \PH$, we use \cref{lem:lipton}.
    Set $\epsilon < \frac{1}{12k}$.
    For $\mixedSigmai[k]$, \cref{lem:lipton} implies that
    \[
        \max_{\rho_1\in U_{t_k}} \min_{\rho_2\in U_{t_{k-1}}} \max_{\rho_3 \in U_{t_{k-3}}} \dots \operatorname*{Q^k}_{\rho_k \in U_{t_1}} \Pr\left[M(\rho_1,\dots,\rho_k)\right] \in [v - k\epsilon, v+k\epsilon] \subseteq \left[v-\frac{1}{12} , v+\frac{1}{12}\right].
    \]
    Given the $\mixedSigmai[k]$ promise gap of $\frac{2}{3},\frac{1}{3}$, this modified game has a promise gap of $\frac{7}{12},\frac{5}{12}$.
    
    We simulate this in $\PH$ as follows. To send the distribution $\rho_i$, the prover sends every string in the support of $\rho_i$, which is only $\poly(n)$ many bits by \cref{lem:lipton}. The verifier can then take the list of strings and sample one uniformly at random. 
    This completes the proof since $\PH$ can simulate randomness \cite{sipser1983complexity, LAUTEMANN1983215}. 
\end{proof}

One can also define $\distributionQCPH$ in the same way, and it follows from \cref{thm:mixedph=ph} that this class is equal to $\QCPH$. 
\begin{corollary}\label{cor:distqcph=qcph}
    $\distributionQCPH = \QCPH$.
\end{corollary}

The \emph{only} difference between $\distributionQCPH$ and $\pureQPH$ is that the former involves proofs that are classical distributions over bit strings and the latter involves proofs that are quantum superpositions over bit strings.
$\distributionQCPH = \QCPH$  is in the counting hierarchy \cite{GSSSY22_qph}, while the best known upper bound for $\pureQPH$ is $\EXP^\PP$ \cite{agarwal2023quantum} and it contains $\QMAtwo$ and $\QPH$.
The conceptual takeaway is that it is only the quantum superposition in the proofs that gives the quantum hierarchies more computational power. 

We also remark that if one allows the distributions in $\distributionPH$ and $\distributionQCPH$ to be correlated, then the techniques in \cref{lem:ent-qph-collapse} can be used to collapse the resulting hierarchies to the second level. The correlated version of $\distributionPH$ collapses to $\StwoP$. 
The correlated version of $\distributionQCPH$ collapses to a quantum-classical version of $\QRGone$, which, to our knowledge, has never been studied.

\section{Open Problems}\label{sec:openproblems}

It is well-known that $\PH$ can equivalently be defined via oracle Turing machines. 
This suggests oracular definitions of quantum polynomial hierarchies, such as $\QMAH$ discussed in \cref{subsec:related-work}. 
One could similarly define $\mathsf{\QCMAH}$ as $\QCMA^{\QCMA^{\QCMA^{\ldots}}}$ and $\mathsf{QMA(2)H}$ as $\QMAtwo^{\QMAtwo^{\QMAtwo^{\ldots}}}$.
We ask how these oracular hierarchies compare to the quantifier-based ones.

\begin{question}\label{question:oracle-vs-quantifier}
    Does $\QEPH = \QMAH$? $\QPH = \mathsf{QMA(2)H}$? $\QCPH = \mathsf{QCMAH}$?
\end{question}

It is unclear if these hierarchies are equal, as in the classical world, or if one version would be stronger than the other.
One immediate obstacle is the fact that $\QEPH$ and $\QPH$ are quantifying over quantum states, so perhaps it is easier to begin with $\QCPH$, which still quantifies over classical bits. 
Alas, it is still unclear if an oracle machine definition of $\QCPH$ would be equal to a quantifier definition, since, in the oracular case, queries can be made in superposition.

Answering \cref{question:oracle-vs-quantifier} could yield progress towards characterizing $\QRGone$.
Jain and Watrous showed that $\QRGone \subseteq \PSPACE$ in 2009 \cite{jain2009parallel}, and, since then, no improved upper bounds have been proven despite effort \cite{ghosh2023complexity}. 
Our work shows that $\QRGone = \QEPH$, and, if one can show $\QEPH \subseteq \QMAH$, then that would imply $\QRGone \subseteq \CH$, since \cite{V18_BQPHthesis} showed that $\QMAH \subseteq \CH$.

More broadly, proving better upper or lower bounds on the quantum polynomial hierarchies and finding more connections to other parts of complexity theory are important directions for future work. 
For example, does any level of $\QPH$ contain $\PSPACE$? 
Can one improve the containment $\QPH \subseteq \EXP^\PP$?
Or, how can these hierarchies be used to better understand the relationships between $\QCMA$, $\QMA$, and $\QMAtwo$?

\section*{Acknowledgements}
We thank Khang Le, Daniel Liang, William Kretschmer, Siddhartha Jain, and Scott Aaronson for helpful conversations. Joshua Cook was especially helpful at early stages of this project. We thank John Watrous for identifying an error in an earlier draft of this work.

\bibliographystyle{alphaurl}
\bibliography{refs}

\newcommand{\etalchar}[1]{$^{#1}$}
\begin{thebibliography}{BCWdW01}

\bibitem[AB09]{arora2009computational}
Sanjeev Arora and Boaz Barak.
\newblock {\em {Computational Complexity: A Modern Approach}}.
\newblock Cambridge University Press, 2009.
\newblock \href {https://doi.org/10.1017/CBO9780511804090}
  {\path{doi:10.1017/CBO9780511804090}}.

\bibitem[AGKR]{agarwal2023quantum}
Avantika Agarwal, Sevag Gharibian, Venkata Koppula, and Dorian Rudolph.
\newblock {Quantum Polynomial Hierarchies: Karp-Lipton, error reduction, and
  lower bounds}.
\newblock To appear.

\bibitem[AIK22]{AIK22_acrobatics}
Scott Aaronson, DeVon Ingram, and William Kretschmer.
\newblock {The Acrobatics of $\mathsf{BQP}$}.
\newblock In {\em 37th Computational Complexity Conference (CCC 2022)}, Leibniz
  International Proceedings in Informatics (LIPIcs), pages 20:1--20:17, 2022.
\newblock \href {https://doi.org/10.4230/LIPIcs.CCC.2022.20}
  {\path{doi:10.4230/LIPIcs.CCC.2022.20}}.

\bibitem[Alt94]{ALTHOFER1994339}
Ingo Althöfer.
\newblock {On sparse approximations to randomized strategies and convex
  combinations}.
\newblock {\em Linear Algebra and its Applications}, 199:339--355, 1994.
\newblock Special Issue Honoring Ingram Olkin.
\newblock \href {https://doi.org/10.1016/0024-3795(94)90357-3}
  {\path{doi:10.1016/0024-3795(94)90357-3}}.

\bibitem[BCWdW01]{buhrman2001quantum}
Harry Buhrman, Richard Cleve, John Watrous, and Ronald de~Wolf.
\newblock Quantum fingerprinting.
\newblock {\em Physical Review Letters}, 87(16):167902, 2001.
\newblock \href {https://doi.org/10.1103/PhysRevLett.87.167902}
  {\path{doi:10.1103/PhysRevLett.87.167902}}.

\bibitem[BFM23]{bassirian2023quantum}
Roozbeh Bassirian, Bill Fefferman, and Kunal Marwaha.
\newblock Quantum {M}erlin-{A}rthur and proofs without relative phase, 2023.
\newblock \href {http://arxiv.org/abs/2306.13247} {\path{arXiv:2306.13247}}.

\bibitem[FGN23]{falor2023collapsible}
Chirag Falor, Shu Ge, and Anand Natarajan.
\newblock {A Collapsible Polynomial Hierarchy for Promise Problems}, 2023.
\newblock \href {http://arxiv.org/abs/2311.12228} {\path{arXiv:2311.12228}}.

\bibitem[FK97]{feige1997making}
Uriel Feige and Joe Kilian.
\newblock {Making Games Short (Extended Abstract)}.
\newblock In {\em Proceedings of the Twenty-Ninth Annual ACM Symposium on
  Theory of Computing}, pages 506--516, New York, NY, USA, 1997. Association
  for Computing Machinery.
\newblock \href {https://doi.org/10.1145/258533.258644}
  {\path{doi:10.1145/258533.258644}}.

\bibitem[FKS95]{feigenbaum1995game}
J.~Feigenbaum, D.~Koller, and P.~Shor.
\newblock {A Game-Theoretic Classification of Interactive Complexity Classes}.
\newblock In {\em Proceedings of Structure in Complexity Theory. Tenth Annual
  IEEE Conference}, pages 227--237, 1995.
\newblock \href {https://doi.org/10.1109/SCT.1995.514861}
  {\path{doi:10.1109/SCT.1995.514861}}.

\bibitem[GK12]{gharibian2012hardness}
Sevag Gharibian and Julia Kempe.
\newblock Hardness of approximation for quantum problems.
\newblock In {\em International Colloquium on Automata, Languages, and
  Programming}, pages 387--398. Springer, 2012.
\newblock \href {https://doi.org/10.1007/978-3-642-31594-7_33}
  {\path{doi:10.1007/978-3-642-31594-7_33}}.

\bibitem[GSS{\etalchar{+}}22]{GSSSY22_qph}
Sevag Gharibian, Miklos Santha, Jamie Sikora, Aarthi Sundaram, and Justin
  Yirka.
\newblock Quantum generalizations of the {P}olynomial {H}ierarchy with
  applications to {$\mathsf{QMA(2)}$}.
\newblock {\em Computational Complexity}, 31(2):13, 2022.
\newblock \href {https://doi.org/10.1007/s00037-022-00231-8}
  {\path{doi:10.1007/s00037-022-00231-8}}.

\bibitem[GW05]{gutoski2005qrg}
Gus Gutoski and John Watrous.
\newblock Quantum interactive proofs with competing provers.
\newblock In {\em STACS 2005: 22nd Annual Symposium on Theoretical Aspects of
  Computer Science}, pages 605--616. Springer, 2005.
\newblock \href {https://doi.org/10.1007/978-3-540-31856-9_50}
  {\path{doi:10.1007/978-3-540-31856-9_50}}.

\bibitem[GW07]{gutoski2007toward}
Gus Gutoski and John Watrous.
\newblock {Toward a General Theory of Quantum Games}.
\newblock In {\em Proceedings of the Thirty-Ninth Annual ACM Symposium on
  Theory of Computing}, pages 565–--574, New York, NY, USA, 2007. Association
  for Computing Machinery.
\newblock \href {https://doi.org/10.1145/1250790.1250873}
  {\path{doi:10.1145/1250790.1250873}}.

\bibitem[GW13]{gutoski2013parallel}
Gus Gutoski and Xiaodi Wu.
\newblock {Parallel Approximation of Min-Max Problems}.
\newblock {\em Computational Complexity}, 22:385--428, 2013.
\newblock \href {https://doi.org/10.1007/s00037-013-0065-9}
  {\path{doi:10.1007/s00037-013-0065-9}}.

\bibitem[GW23]{ghosh2023complexity}
Soumik Ghosh and John Watrous.
\newblock Complexity limitations on one-turn quantum refereed games.
\newblock {\em Theory of Computing Systems}, 67(2):383--412, 2023.
\newblock \href {https://doi.org/10.1007/s00224-022-10105-9}
  {\path{doi:10.1007/s00224-022-10105-9}}.

\bibitem[HM13]{harrow2013testing}
Aram~W. Harrow and Ashley Montanaro.
\newblock {Testing Product States, Quantum Merlin-Arthur Games and Tensor
  Optimization}.
\newblock {\em J. ACM}, 60(1), 2013.
\newblock \href {https://doi.org/10.1145/2432622.2432625}
  {\path{doi:10.1145/2432622.2432625}}.

\bibitem[JJUW11]{jain2011qip}
Rahul Jain, Zhengfeng Ji, Sarvagya Upadhyay, and John Watrous.
\newblock {$\mathsf{QIP} = \mathsf{PSPACE}$}.
\newblock {\em Journal of the ACM (JACM)}, 58(6):1--27, 2011.
\newblock \href {https://doi.org/10.1145/2049697.2049704}
  {\path{doi:10.1145/2049697.2049704}}.

\bibitem[JW09]{jain2009parallel}
Rahul Jain and John Watrous.
\newblock {Parallel Approximation of Non-interactive Zero-sum Quantum Games}.
\newblock In {\em 24th Annual IEEE Conference on Computational Complexity},
  pages 243--253, 2009.
\newblock \href {https://doi.org/10.1109/CCC.2009.26}
  {\path{doi:10.1109/CCC.2009.26}}.

\bibitem[JW23]{jeronimo2023power}
Fernando~Granha Jeronimo and Pei Wu.
\newblock The power of unentangled quantum proofs with non-negative amplitudes.
\newblock In {\em Proceedings of the 55th Annual ACM Symposium on Theory of
  Computing}, STOC 2023, pages 1629–--1642, New York, NY, USA, 2023.
  Association for Computing Machinery.
\newblock \href {https://doi.org/10.1145/3564246.3585248}
  {\path{doi:10.1145/3564246.3585248}}.

\bibitem[KKMP21]{kleinberg_et_al:LIPIcs.ITCS.2021.44}
Robert Kleinberg, Oliver Korten, Daniel Mitropolsky, and Christos
  Papadimitriou.
\newblock {Total Functions in the Polynomial Hierarchy}.
\newblock In {\em 12th Innovations in Theoretical Computer Science Conference
  (ITCS 2021)}, volume 185 of {\em Leibniz International Proceedings in
  Informatics (LIPIcs)}, pages 44:1--44:18, 2021.
\newblock \href {https://doi.org/10.4230/LIPIcs.ITCS.2021.44}
  {\path{doi:10.4230/LIPIcs.ITCS.2021.44}}.

\bibitem[KSV02]{kitaev2002classical}
Alexei~Y. Kitaev, Alexander Shen, and Mikhail~N. Vyalyi.
\newblock {\em {Classical and Quantum Computation}}.
\newblock American Mathematical Soc., 2002.
\newblock \href {https://doi.org/10.1090/gsm/047} {\path{doi:10.1090/gsm/047}}.

\bibitem[KW00]{kitaev2000parallelization}
Alexei Kitaev and John Watrous.
\newblock {Parallelization, amplification, and exponential time simulation of
  quantum interactive proof systems}.
\newblock In {\em Proceedings of the 32nd Annual ACM Symposium on Theory of
  Computing}, pages 608--617, 2000.
\newblock \href {https://doi.org/10.1145/335305.335387}
  {\path{doi:10.1145/335305.335387}}.

\bibitem[Lau83]{LAUTEMANN1983215}
Clemens Lautemann.
\newblock {$\BPP$ and the polynomial hierarchy}.
\newblock {\em Information Processing Letters}, 17(4):215--217, 1983.
\newblock \href {https://doi.org/10.1016/0020-0190(83)90044-3}
  {\path{doi:10.1016/0020-0190(83)90044-3}}.

\bibitem[LY94]{lipton1994simple}
Richard~J. Lipton and Neal~E. Young.
\newblock {Simple Strategies for Large Zero-Sum Games with Applications to
  Complexity Theory}.
\newblock In {\em Proceedings of the 26th Annual ACM Symposium on Theory of
  Computing}, pages 734--740, 1994.
\newblock \href {https://doi.org/10.1145/195058.195447}
  {\path{doi:10.1145/195058.195447}}.

\bibitem[MS72]{meyer1972equivalence}
A.~R. Meyer and L.~J. Stockmeyer.
\newblock The equivalence problem for regular expressions with squaring
  requires exponential space.
\newblock In {\em 13th Annual Symposium on Switching and Automata Theory (SWAT
  1972)}, pages 125--129, 1972.
\newblock \href {https://doi.org/10.1109/SWAT.1972.29}
  {\path{doi:10.1109/SWAT.1972.29}}.

\bibitem[MW05]{marriott2005quantum}
Chris Marriott and John Watrous.
\newblock {Quantum Arthur--Merlin Games}.
\newblock {\em Computational Complexity}, 14(2):122--152, 2005.
\newblock \href {https://doi.org/10.1007/s00037-005-0194-x}
  {\path{doi:10.1007/s00037-005-0194-x}}.

\bibitem[NC10]{nielsen2002quantum}
Michael~A. Nielsen and Isaac~L. Chuang.
\newblock {\em {Quantum Computation and Quantum Information: 10th Anniversary
  Edition}}.
\newblock Cambridge University Press, 2010.
\newblock \href {https://doi.org/10.1017/CBO9780511976667}
  {\path{doi:10.1017/CBO9780511976667}}.

\bibitem[Sha92]{shamir1992ip}
Adi Shamir.
\newblock {$\mathsf{IP} = \mathsf{PSPACE}$}.
\newblock {\em J. ACM}, 39(4):869–877, 1992.
\newblock \href {https://doi.org/10.1145/146585.146609}
  {\path{doi:10.1145/146585.146609}}.

\bibitem[Sio58]{sion1958general}
Maurice Sion.
\newblock {On General Minimax Theorems}.
\newblock {\em Pacific Journal of Mathematics}, 1958.
\newblock \href {https://doi.org/10.2140/pjm.1958.8.171}
  {\path{doi:10.2140/pjm.1958.8.171}}.

\bibitem[Sip83]{sipser1983complexity}
Michael Sipser.
\newblock {A Complexity Theoretic Approach to Randomness}.
\newblock In {\em Proceedings of the Fifteenth Annual ACM Symposium on Theory
  of Computing}, pages 330--335. Association for Computing Machinery, 1983.
\newblock \href {https://doi.org/10.1145/800061.808762}
  {\path{doi:10.1145/800061.808762}}.

\bibitem[Sto76]{stockmeyer1976polynomial}
Larry~J. Stockmeyer.
\newblock {The Polynomial-Time Hierarchy}.
\newblock {\em Theoretical Computer Science}, 3(1):1--22, 1976.
\newblock \href {https://doi.org/10.1016/0304-3975(76)90061-X}
  {\path{doi:10.1016/0304-3975(76)90061-X}}.

\bibitem[Vin18]{V18_BQPHthesis}
Lieuwe Vinkhuijzen.
\newblock {A Quantum Polynomial Hierarchy and a Simple Proof of Vyalyi’s
  Theorem}.
\newblock Master's thesis, Leiden University, 2018.
\newblock URL: \url{https://theses.liacs.nl/1505}.

\bibitem[Wat18]{watrous2018theory}
John Watrous.
\newblock {\em {The Theory of Quantum Information}}.
\newblock Cambridge University Press, 2018.
\newblock \href {https://doi.org/10.1017/9781316848142}
  {\path{doi:10.1017/9781316848142}}.

\bibitem[Yam02]{yamakami2002quantum}
Tomoyuki Yamakami.
\newblock {Quantum $\NP$ and a Quantum Hierarchy}.
\newblock In {\em Foundations of Information Technology in the Era of
  Networking and Mobile Computing}, volume~96 of {\em IFIP — The
  International Federation for Information Processing}, pages 323--336, Boston,
  MA, 2002. Springer.
\newblock \href {https://doi.org/10.1007/978-0-387-35608-2_27}
  {\path{doi:10.1007/978-0-387-35608-2_27}}.

\end{thebibliography}

\end{document}